\documentclass[11pt,xcolor=dvipsnames, reqno]{amsart}
\usepackage[a4paper, top=3cm, bottom=4cm, left=3cm, right=3cm]{geometry}
\usepackage[T1]{fontenc}
\usepackage[utf8]{inputenc}
\usepackage[english]{babel}
\usepackage{amsmath,amssymb}
\usepackage{amsthm}
\usepackage{amsfonts}
\usepackage{amsxtra}
\usepackage{graphicx}
\usepackage{wrapfig}
\usepackage[dvipsnames]{xcolor}
\usepackage{subcaption}

\usepackage{hyperref}
\usepackage[backend=biber,style=numeric,sorting=none,firstinits=true, maxbibnames=99]{biblatex}

\allowdisplaybreaks

\bibliography{References.bib}

\newcommand{\R}{\mathbb{R}}
\newcommand{\N}{\mathbb{N}}
\newcommand{\C}{\mathbb{C}}

\newcommand{\rd}{\,\mathrm{d}}
\newcommand{\dr}[1]{\mathrm{d}{#1}\, }

\newcommand{\ri}{\mathrm{i}}

\newcommand{\im}{\operatorname{Im}}

\let\epsilon\varepsilon

\newtheorem{theorem}{Theorem}[section]
\newtheorem{lemma}[theorem]{Lemma}
\newtheorem{definition}[theorem]{Definition}
\newtheorem{corollary}[theorem]{Corollary}
\newtheorem{remark}[theorem]{Remark}

\title[Dispersive Estimates for NLS with External Potentials]{Dispersive Estimates for Nonlinear Schrödinger Equations with External Potentials}
\author{Charlotte Dietze}
\address{Department of Mathematics, LMU Munich, Theresienstr. 39, 80333 Munich, Germany}
\email{dietze@math.lmu.de}
\date{\today}

\makeatletter
\@namedef{subjclassname@2020}{\textup{2020} Mathematics Subject Classification}
\makeatother

\subjclass[2020]{35Q55}

\begin{document}
\maketitle
\begin{abstract}
We consider the long time dynamics of nonlinear Schrödinger equations with an external potential. More precisely, we look at Hartree type equations in three or higher dimensions with small initial data. We prove an optimal decay estimate, which is comparable to the decay of free solutions. Our proof relies on good control on a high Sobolev norm of the solution to estimate the terms in Duhamel's formula.
\end{abstract}
\tableofcontents
\section{Introduction}
Nonlinear Schrödinger equations are of great interest in physics and mathematics, see \cite{malomedinscott, cazenave, dauxoispeyrard} for an overview. They are used to model waves on the surface of a deep fluid, see \cite{zakharov}, to describe Langmuir waves in plasma physics, see \cite{goldman}, and they are used in nonlinear optics, see \cite{schneider}. Furthermore, these equations describe intramolecular vibrations in $\alpha$-helices in proteins, see \cite{davydov3, dauxoispeyrard}. The condensate in Bose-Einstein condensation can also be described by nonlinear Schrödinger equations via the mean-field approximation, which are usually called \emph{Hartree} or \emph{Gross-Pitaevskii equations}, see \cite{gross, pitaevskii, leggett, pomeaurica}. 
\\ \\
We consider the nonlinear Schrödinger equation with a Hartree type nonlinearity
\begin{equation} \label{eq:Hartree type equation with external potential introduction}
\begin{cases} 
\ri\partial_t u & =-\Delta u+ Vu+(w*|u|^2)u\\
u(0) & =u_0
\end{cases}
\end{equation}
in dimension $d\ge 3$ with an external potential $V:\R^d\to\R$ and an interaction potential $w:\R^d\to\R$. We call this equation a \emph{Hartree type equation}. For $V,w$ nice enough and small initial data $u_0$, we show the decay estimate 
\begin{equation}\label{eq:general decay estimate with 1+t}
||u(t)||_\infty\le \frac{C}{(1+|t|)^{\frac{d}{2}}}
\end{equation}
for a constant $C>0$, where $C$ depends on $w$ only in terms of $||w||_1$. For $p\in [1,\infty]$, we denote by $||\cdot||_p$ the $L^p$ norm on $\R^d$. Such a decay estimate was proved in \cite[Corollary 3.4]{gm} for $d=3$, $V=0$ and even large initial data. The decay estimate from \cite{gm} was used in \cite{gm,namnapiorkowski} to understand the dynamics of many-body quantum systems in the context of Bose-Einstein condensation in dimension $d=3$ without an external potential. It should be possible to use our decay estimate to get similar results in the weak coupling regime in dimension $d\ge 3$ with an external potential $V$.
\\ \\
{\bf Linear equation.}
To get a better understanding why we might expect a decay estimate of the form \eqref{eq:general decay estimate with 1+t} for nonlinear Schrödinger equations under certain conditions, let us look at the linear Schrödinger equation
\begin{equation}
\begin{cases}
\ri\partial_t u & =-\Delta u\\
u(0) & =u_0 
\end{cases}
\end{equation}
with initial data $u_0\in L^1(\R^d)\cap L^2(\R^d)$. The solution to this equation, see for example \cite[Equation (4.2)]{nd}, is given by
\begin{equation}
u(t,x)=\left(e^{-\ri t(-\Delta)}u_0\right)(x)=\frac{1}{(4\pi\ri t)^{\frac{d}{2}}}\int_{\R^d}\dr y e^{\frac{\ri |x-y|^2}{4t}}u_0(y)\, . 
\end{equation}
By taking the $L^\infty$ norm on both sides, we obtain the decay estimate
\begin{equation}\label{eq:dispersive estimate for e^(it Laplace)}
||u(t)||_\infty\le\frac{1}{(4\pi)^{\frac{d}{2}}}\frac{1}{|t|^{\frac{d}{2}}}||u_0||_1\, . 
\end{equation}
This decay rate agrees with the decay rate in equation \eqref{eq:general decay estimate with 1+t} for large times $t$. Note that the estimate \eqref{eq:general decay estimate with 1+t} is even stronger than \eqref{eq:dispersive estimate for e^(it Laplace)} for small times $t$, which will follow from our assumptions on the initial data. 
\\ \\
{\bf Nonlinear Schrödinger equations without external potentials.}
The equation \eqref{eq:Hartree type equation with external potential introduction} with $V=0$, namely 
\begin{equation}
\begin{cases} 
\ri\partial_t u & =-\Delta u+(w*|u|^2)u\\
u(0) & =u_0
\end{cases}
\end{equation}
has been studied extensively in the literature. There are many results on local or global existence and uniqueness, scattering, modified scattering and wave operators for these equations. Strauss studied scattering theory in a general setting and applied it to nonlinear Schrödinger equations, see \cite{straussscattering2,straussscattering3}. For Hartree type nonlinearities with repulsive interaction potential $w$, Ginibre and Velo proved the decay estimate
\begin{equation}
||u(t)||_q\le C(1+|t|)^{-\frac{d}{2}\left(\frac{1}{q'}-\frac{1}{q}\right)}
\end{equation}  
for all $q$ such that $\left[\frac{1}{2}-\frac{1}{d}\right]_+\le\frac{1}{q}\le\frac{1}{2}$, see  \cite[Theorem 6.1(1)]{ginibrevelononlocal}. Note that $q=\infty$ is not allowed here if $d\ge 3$. We would also like to mention  \cite{ginibrevelotimedecaykleingordonnls, hayashitsutsumi, cazenaveweissler, ozawa, ginibreozawa, ginibreveloscatteringhartree}.
More recent results include \cite{deiftzhou, duyckaertsholmerroudenko, choozawa}.
\\ \\
Let us now discuss some results and the corresponding proof ideas, which will be important for the proof of our main result. 
\\ \\
Hayashi and Naumkin were the first to prove a decay estimate of the form (\ref{eq:general decay estimate with 1+t}) for critical nonlinearities and small initial data, see \cite{hayashinaumkin}. The nonlinearities they considered were
\begin{equation}
\lambda |u|^{\frac{2}{d}}u+\mu |u|^{\eta -1}u
\end{equation} 
in $d\in\{1,2,3\}$ with $\eta-1>\frac{2}{d}$, $\lambda, \mu\in\R$ and 
\begin{equation}
\lambda (|\cdot|^{-1}*|u|^2)u+\mu (|\cdot|^{-\delta}*|u|^2)u
\end{equation}
in $d\ge 2$ with $1<\delta <d$, $\lambda, \mu\in\R$. Moreover, they proved modified scattering for these equations. 
\\ \\
Kato and Pusateri provided an alternative proof of the result in \cite{hayashinaumkin} for the local nonlinearity $\pm |u|^2u$ in $d=1$ and the Hartree type nonlinearity $\pm (|\cdot|^{-1}*|u|^2)u$ in $d\ge 2$ in \cite{katopusateri}. They defined a quantity $||u||_{X_T}$ depending on a time parameter $T\ge 0$ and they proved an estimate of the form $||u||_{X_T}\le\epsilon+C||u||_{X_T}^3$, where $\epsilon>0$ is small and both $\epsilon, C$ are independent of $T$. They used this inequality to deduce that $\sup_{T\ge 0}||u||_{X_T}<\infty$ for small initial data. Their proof relied on a careful analysis of the equation in the Fourier space.
\\ \\
For the Hartree type equation with non-negative, spherically symmetric and decreasing $w\in L^1(\R^d)\cap C^1_0(\R^d)$, Grillakis and Machedon showed a decay estimate of the form (\ref{eq:general decay estimate with 1+t}) for initial data $u_0\in W^{k,1}(\R^d)$ for $k$ sufficiently large, see \cite[Corollary 3.4]{gm}. It is worth pointing out that their result holds for large initial data. Their result was applied in \cite{gm, namnapiorkowski} in the context of Bose-Einstein condensation to show a norm approximation for the dynamics. Another more recent result on the dynamics of many-body quantum systems, which we would like to mention, is \cite{brenneckenamnapiorkowskischlein}.
\\ \\
{\bf Nonlinear Schrödinger equations with external potentials.}
We start by looking at results in dimension $d=1$. Cuccagna, Georgiev and Visciglia proved a decay estimate of the form (\ref{eq:general decay estimate with 1+t}) and scattering for small initial data and a nonlinearity of the form $\pm |u|^{p-1}u$ for $3<p<5$, see \cite{cuccagnageorgievvisciglia}. Germain, Pusateri and Rousset considered the nonlinearity $\pm |u|^{p-1}u$ with $p=3$, see \cite{germainpusaterirousset}. They proved a decay estimate of the form (\ref{eq:general decay estimate with 1+t}) and modified scattering for small initial data. In their proof, they used the distorted Fourier transform and they carefully analysed an oscillatory integral. 
\\ \\
Naumkin considered the cubic nonlinear Schrödinger equation with an external potential and proved a decay estimate of the form (\ref{eq:general decay estimate with 1+t}) and modified scattering using the distorted Fourier transform, see \cite{naumkinpotential,naumkinexceptional}.
\\ \\
Martinez proved decay in the sense that $\lim_{n\to\infty} ||u(t)||_{L^\infty(I)}=0$ for any bounded interval $I\subset\R$ for small odd solutions $u$ to nonlinear Schrödinger equations with even external potentials $V$, see \cite{martinez}. The nonlinearities considered in \cite{martinez} are of the form $f(|u|^2)u$ for a function $f:\R\to\R$ with $|f(s)|\lesssim s^{\frac{p-1}{2}}$ for $1<p<5$ and Hartree type nonlinearities $\pm \left(|\cdot|^{-\alpha}*|u|^2\right)u$ for $0<\alpha<1$. 
\\ \\
Let us now mention some results in dimension $d=3$. Pusateri and Soffer proved a decay estimate of the form $||u(t)||_\infty\le C(1+|t|)^{-(1+\alpha)}$ for some $\alpha>0$ for the nonlinear Schrödinger equation with nonlinearity $-u^2$ and small initial data, see \cite{pusaterisoffer}. In a similar way to \cite{germainpusaterirousset}, the proof in \cite{pusaterisoffer} relies on the distorted Fourier transform and a careful analysis of an oscillatory integral.
\\ \\
Hong proved scattering in $H^1$ for the cubic focusing nonlinear Schrödinger equation with an external potential $V$ with small negative part, see \cite{hong}. Hong's proof strategy was to show that there are no minimal blow-up solutions. Nakanishi classified the dynamics of solutions to the cubic nonlinear Schrödinger equation with small initial data and a radial external potential $V$, which is such that the operator $-\Delta +V$ has exactly one negative eigenvalue, see \cite{nakanishi}.
\\ \\
There are only few decay results for nonlinear Schrödinger equations with external potentials in $d=3$. Note that the results in $d=3$, which we mentioned here, treat local nonlinearities. Nonlinear Schrödinger equations with non-local nonlinearities such as the Hartree type equation, which we consider below, are not necessarily easier to deal with.
\subsection{Main result}
Our main result is a decay estimate of the form (\ref{eq:general decay estimate with 1+t}) for the Hartree type equation with small initial data in dimension $d\ge 3$. 

\begin{theorem}[Dispersive estimates for the Hartree type equation in $d\ge 3$ for small initial data]\label{th:dispersive d>=3 small}
Let $d\ge 3$ and let $k\in\N$ be the smallest even number with $k>\frac{d}{2}$. Let $V\in W^{k,\infty}(\R^d)$ be a real-valued function and satisfy
\begin{equation}
||e^{-\ri t(-\Delta+V)}f||_\infty\le C^V |t|^{-\frac{d}{2}}||f||_1
\end{equation} 
for every $f\in L^1(\R^d)\cap L^2(\R^d)$ and some constant $C^V\ge 1$. Let the interaction potential $w\in L^1(\R^d)\cap L^{\frac{d}{2}}(\R^d)$ be an even, real-valued function. Let $u_0\in H^k(\R^d)$ and let $u\in C\left(\R , H^2(\R^d)\right)\cap C\left(\R, H^{-1}(\R^d)\right)$ be the unique global strong solution to the Hartree type equation 
\begin{equation}
\begin{cases}
\ri\partial_t u & =(-\Delta +V)u+(w*|u|^2)u\\
u(0) & =u_0
\end{cases}
\end{equation}
given by Theorem \ref{th:global H^2 solutions}. Assume that the initial data is sufficiently small, that is, 
\begin{equation}\label{eq:initial data are small d>=3 small}
||e^{\ri (-\Delta+V)}u_0||_1\, , ||u_0||_{H^k}\le\epsilon_0
\end{equation}
for some $\epsilon_0=\epsilon_0(d,V,||w||_1)>0$. Then there exists a constant $C_0=C_0(d,V,||w||_1)>0$ such that
\begin{equation}\label{eq:dispersive estimate for u main theorem}
||u(t)||_\infty\le\frac{C_0}{(1+|t|)^{\frac{d}{2}}}
\end{equation}
for all $t\ge 0$. Furthermore, if we assume that
\begin{equation}
||e^{\ri (-\Delta+V)}(\partial_t u)(0)||_1\, ,||(\partial_t u)(0)||_{H^k}\le\tilde\epsilon_0
\end{equation}
for some $\tilde\epsilon_0=\tilde \epsilon_0(d,V,||w||_1)>0$, then
\begin{equation}
||\partial_t u(t)||_\infty\le  \frac{\tilde C_0}{(1+|t|)^{\frac{d}{2}}}
\end{equation}
for all $t\ge 0$, where $\tilde C_0=\tilde C_0(d,V,||w||_1)>0$. 
\end{theorem}
This result is new even when $d=3$. 
\begin{remark}
A special case of Theorem \ref{th:dispersive d>=3 small} was proved by Grillakis and Machedon \cite[Corollary 3.4]{gm}, where they considered $V=0$ in dimension $d=3$. Their result even holds for large initial data. 
\end{remark}
\begin{remark}[Application in many-body quantum mechanics]\label{re:application in many-body QM}
Note that the $w$-dependence of the constants $\epsilon_0$ and $C_0$ is only in terms of $||w||_1$. For proving the norm approximation for the dynamics of many-body quantum systems in dimension $d=3$ without an external potential in \cite{gm, namnapiorkowski}, it was crucial that the constants in the decay estimate in \cite[Corollary 3.4]{gm} depended on $w$ only in terms of $||w||_1$. Using Theorem \ref{th:dispersive d>=3 small}, it should be possible to prove analogous results in the small coupling regime for the many-body Schrödinger equation with an external potential $V$ in dimension $d\ge 3$. 
\end{remark}
\begin{remark}[Dispersive estimates for $e^{-\ri t (-\Delta+V)}$]
In Theorem \ref{th:dispersive estimate for e^(-itH)}, we mention two different conditions under which a dispersive estimate of the form $||e^{-\ri t(-\Delta+V)}f||_\infty\le C^V |t|^{-\frac{d}{2}}||f||_1$ is satisfied.
\end{remark}
\begin{remark}[Extensions of Theorem \ref{th:dispersive d>=3 small}] 
\begin{itemize}
\item[(i)] \sloppy A similar decay result holds for $t\le 0$. For that case, we need to replace the condition $||e^{\ri (-\Delta+V)}u_0||_1\le\epsilon_0$ by $||e^{-\ri (-\Delta+V)}u_0||_1\le\epsilon_0$ and $||e^{\ri (-\Delta+V)}(\partial_t u)(0)||_1\le\tilde\epsilon_0$ by $||e^{-\ri (-\Delta+V)}(\partial_t u)(0)||_1\le\tilde\epsilon_0$.
\item[(ii)] In Remark \ref{re:large d>=3}, we treat the case of large initial data, which will be proved using Gronwall's lemma. However, we will need the additional assumptions that $\lim_{t\to\infty}||u(t)||_\infty=0$ and $\sup_{t\ge 0}||D^ku(t)||_\infty <\infty$. 
\end{itemize}
\end{remark}
\begin{remark}[Further questions]
\begin{itemize}
\item[(i)] It would be interesting to consider the Hartree equation with an external potential $V\neq 0$, where the interaction potential $w$ is given by $w(x)=\frac{1}{|x|}$. The proof of a decay estimate of the form \eqref{eq:dispersive estimate for u main theorem} for the Hartree equation with $V=0$ in \cite{katopusateri} relied on a careful analysis in Fourier space. Similar to \cite{pusaterisoffer}, one could try to use an approach involving a distorted Fourier transform for the Hartree equation with an external potential $V\neq 0$. 
\item[(ii)]An even more challenging problem is to understand the dynamics of solutions to the Hartree equation for large initial data with an external potential $V(x)=-\frac{Z}{|x|}$, where $Z>0$, which is closely linked to the dynamical ionisation conjecture, see \cite{lenzmannlewin}. 
\end{itemize}
\end{remark}
An analogous result to Theorem \ref{th:dispersive d>=3 small} for the cubic nonlinear Schrödinger equation is the following theorem. 
\begin{theorem}[Dispersive estimates for the cubic nonlinear Schrödinger equation for small initial data]\label{th:cubic NLS main result}
Let $d\ge 3$ and suppose that the assumptions on the external potential $V$ and the smallness assumptions on the initial data $u_0$ from Theorem \ref{th:dispersive d>=3 small} are satisfied. Let $u$ be a global solution to either the focusing or the defocusing cubic nonlinear Schrödinger equation
\begin{equation} 
\begin{cases} 
\ri\partial_t u & =-\Delta u+ Vu\pm |u|^2u\\
u(0) & =u_0 \, .
\end{cases}
\end{equation}
Then the dispersive estimates
\begin{equation}
||u(t)||_\infty\le\frac{C_0}{(1+|t|)^{\frac{d}{2}}}
\end{equation}
and
\begin{equation}
||\partial_t u(t)||_\infty\le\frac{\tilde C_0}{(1+|t|)^{\frac{d}{2}}}
\end{equation}
hold for all for all $t\ge 0$ and some constants $C_0=C_0(d,V)>0$ and $\tilde C_0=\tilde C_0(d,V)>0$.
\end{theorem}
\begin{remark}[The cubic nonlinear Schrödinger equation as a limit of Hartree type equations]
We can view the cubic nonlinear Schrödinger equation as a limit of Hartree type equations with interaction potentials $w$ that converge in the distributional sense to the delta distribution $\delta_0$ or $-\delta_0$, respectively. Moreover, the $w$-dependence of the constant $C_0$ in Theorem \ref{th:dispersive d>=3 small} is only in terms of $||w||_1$. For these two reasons, it is natural to expect that a result such as Theorem \ref{th:cubic NLS main result} holds if Theorem \ref{th:dispersive d>=3 small} is true, which is the corresponding result for Hartree type equations. At the end of Section \ref{s:conclusion of the main theorem}, we explain two different strategies for proving Theorem \ref{th:cubic NLS main result}.  
\end{remark}
\subsection{Proof strategy}\label{ss:proof strategy introduction}
In this subsection, we describe our proof strategy for the estimate
\begin{equation}
||u(t)||_\infty\le\frac{C_0}{(1+|t|)^{\frac{d}{2}}}
\end{equation}
in Theorem \ref{th:dispersive d>=3 small}. The proof idea is a combination of the proof ideas in \cite{gm} and \cite{katopusateri}. The constants $C>0$ in this subsection can change from line to line but they do not depend on $t$ or $u$. Define $H:=-\Delta+V$. 
\\ \\
By Duhamel's formula, see Lemma \ref{le:Duhamel}, we have
\begin{equation}
u(t) = e^{-\ri tH}u_0-\ri\int_0^t\dr s e^{-\ri (t-s)H}(w*|u(s)|^2)u(s)\, .
\end{equation}
Taking the $L^\infty$ norm on both sides and using the dispersive estimate for $e^{-\ri tH}$ of the form $||e^{-\ri tH}f||_\infty\le C|t|^{-\frac{d}{2}}||f||_1$, we get for $t\ge 0$
\begin{align*}
||u(t)||_\infty &\le ||e^{-\ri tH}u_0||_\infty+\int_0^t\dr s ||e^{-\ri (t-s)H}(w*|u(s)|^2)u(s)||_\infty\\
&\le ||e^{-\ri (t+1)H}e^{\ri H}u_0||_\infty+ \int_0^t\dr s C|t-s|^{-\frac{d}{2}}||(w*|u(s)|^2)u(s)||_1\\
&\le C|t+1|^{-\frac{d}{2}}||e^{\ri H}u_0||_1+ C\int_0^t\dr s |t-s|^{-\frac{d}{2}}||w||_1||u(s)||_2^2||u(s)||_\infty\\
&\le C(1+|t|)^{-\frac{d}{2}}||e^{\ri H}u_0||_1+ C||w||_1||u_0||_2^2\int_0^t\dr s |t-s|^{-\frac{d}{2}}||u(s)||_\infty\, ,
\end{align*}
where we used Young's inequality and Hölder's inequality in the second last step and the conservation of the $L^2$ norm of $u$ (see Theorem \ref{th:local well-posedness and conservation of mass and energy}) in the last step. Note that $|\cdot|^{-\frac{d}{2}}\notin L^1(0,1)$ for $d\ge 3$, so the last integral is infinite unless $\liminf_{s\uparrow t}||u(s)||_\infty= 0$. Therefore, we should estimate the integral for s close to $t$ differently. More precisely, we will use a different estimate for $s\in [t_0,t]$, where $ t_0:=\max\{t-1,0\}$. Call this term
\begin{equation}
(R):=\int_{t_0}^t\dr s ||e^{-\ri (t-s)H}(w*|u(s)|^2)u(s)||_\infty\, .
\end{equation}
We will explain how to estimate $(R)$ later. So far, we have shown that
\begin{equation}\label{eq:a first estimate for ||u(t)||_infty using the direct estimate with a term (R)}
||u(t)||_\infty\le C(1+|t|)^{-\frac{d}{2}}||e^{\ri H}u_0||_1+ C||w||_1||u_0||_2^2\int_0^{t_0}\dr s |t-s|^{-\frac{d}{2}}||u(s)||_\infty +( R)\, . 
\end{equation}
An estimate of the form (\ref{eq:general decay estimate with 1+t}) for all $t\ge 0$ can also be written as
\begin{equation}\label{eq:dispersive estimate written as sup |t|^(d/2) ||u(t)||_infty<=C}
\sup_{t\ge 0} (1+|t|)^{\frac{d}{2}}||u(t)||_\infty\le C<\infty
\end{equation}
for some constant $C>0$. If we define
\begin{equation}
N(T):=\sup_{0\le t\le T} (1+|t|)^{\frac{d}{2}}||u(t)||_\infty\, ,
\end{equation}
then (\ref{eq:dispersive estimate written as sup |t|^(d/2) ||u(t)||_infty<=C}) is equivalent to 
\begin{equation}
N(T)\le C
\end{equation}
for all $T\ge 0$, where the constant $C>0$ is independent of $T$. By (\ref{eq:a first estimate for ||u(t)||_infty using the direct estimate with a term (R)}), we have
\begin{align*}
&\quad (1+|t|)^{\frac{d}{2}}||u(t)||_\infty\\
&\le C||e^{\ri H}u_0||_1+ C(1+|t|)^{\frac{d}{2}}||w||_1||u_0||_2^2\int_0^{t_0}\dr s |t-s|^{-\frac{d}{2}}(1+|s|)^{-\frac{d}{2}}N(s)+(1+|t|)^{\frac{d}{2}}(R)\, .
\end{align*}
Let $0\le t\le T$. Since $N(s)\le N(T)$ for $0\le s\le t$, we get
\begin{align*}
&\quad (1+|t|)^{\frac{d}{2}}||u(t)||_\infty\\
&\le C||e^{\ri H}u_0||_1+ C||w||_1||u_0||_2^2N(T)\int_0^{t_0}\dr s (1+|t|)^{\frac{d}{2}}|t-s|^{-\frac{d}{2}}(1+|s|)^{-\frac{d}{2}}+ (1+|t|)^{\frac{d}{2}}(R)\, .
\end{align*}
Note that by the definition of $t_0$, there exists a constant $C>0$ independent of $t\ge 0$ such that
\begin{equation}
\int_0^{t_0}\dr s (1+|t|)^{\frac{d}{2}}|t-s|^{-\frac{d}{2}}(1+|s|)^{-\frac{d}{2}}\le C\, .
\end{equation}
We can see this by splitting the integral into two terms: The first term is the integral over $s\in [0,t_0]$ with $1+s\le\frac{1+t}{2}$, where we estimate $|t-s|^{-\frac{d}{2}}=|(1+t)-(1+s)|^{-\frac{d}{2}}\le \left|\frac{1+t}{2}\right|^{-\frac{d}{2}}$. Thus, we can estimate the integrand by $2^{\frac{d}{2}}(1+s)^{-\frac{d}{2}}\in L^1([0,\infty))$. Similarly, the second term is the integral over $s\in[0,t_0]$ with $1+s\ge\frac{1+t}{2}$. Therefore, we can estimate the integrand by $2^{\frac{d}{2}}|t-s|^{-\frac{d}{2}}$. Since $t_0=\max\{0,t-1\}$, we can estimate the second term by $2^{\frac{d}{2}}\int_1^\infty \dr r |r|^{-\frac{d}{2}}<\infty$. Note that we used $d\ge 3$ here. 
\\ \\
We get
\begin{equation}\label{eq:estimate for (1+t)^(d/2) infinity norm of u by an expression with N(T) and (R)}
(1+|t|)^{\frac{d}{2}}||u(t)||_\infty \le C_1||e^{\ri H}u_0||_1+ C_2N(T)||w||_1||u_0||_2^2+ (1+|t|)^{\frac{d}{2}}(R)
\end{equation}
for some constants $C_1,C_2>0$. Let us forget about $(R)$ for a moment; that is,  assume $(R)=0$. If 
\begin{equation}\label{eq:L2 norm of u_0 small introduction}
C_2||w||_1||u_0||_2^2\le\frac{1}{2}\, , 
\end{equation}
then we obtain
\begin{equation}\label{eq:N(T)<=C/2+1/2N(T) introduction}
N(T)\le \frac{\tilde C}{2}+\frac{1}{2}N(T)
\end{equation}
with $\tilde C:=2C_1||e^{\ri H}u_0||_1$ by taking the supremum over $0\le t\le T$. Note that \eqref{eq:L2 norm of u_0 small introduction} is satisfied when $||u_0||_2$ is small enough. An inequality such as (\ref{eq:N(T)<=C/2+1/2N(T) introduction}) would then imply that $N(T)\le \tilde C$ for every $T\ge 0$ if we know that $N(T)<\infty$ for every $T<\infty$. 
\\ \\
However, it is not that simple. We still have to estimate $(R)$, which is the most difficult term in the proof of Theorem \ref{th:dispersive d>=3 small}. For the estimate of $(R)$, we will need good control on $\sup_{0\le t\le T}||D^k u(t)||_2$, where $k$ is the even integer with $k>\frac{d}{2}$ from the assumptions of Theorem \ref{th:dispersive d>=3 small}. Therefore, instead of considering $N(T)$, it turns out to be more helpful to look at 
\begin{equation}
M(T):= \sup_{0\le t\le T} (1+|t|)^{\frac{d}{2}}||u(t)||_\infty+\sup_{0\le t \le T}||D^k u(t)||_2+||u_0||_2\, . 
\end{equation}
Note that the definition of $M(T)$ is similar to the definition of $||u||_{X_T}$ in \cite{katopusateri}. Moreover, by (\ref{eq:estimate for (1+t)^(d/2) infinity norm of u by an expression with N(T) and (R)}), we have
\begin{equation}\label{eq:Estimate for (1+t)^(d/2) infinity norm of u with M(T) and (R)}
(1+|t|)^{\frac{d}{2}}||u(t)||_\infty \le C_1||e^{\ri H}u_0||_1+ C_2||w||_1M(T)^3+ (1+|t|)^{\frac{d}{2}}(R)\, .
\end{equation}
We will show that there exists a constant $C_0>0$ such that $M(T)\le C_0$ for all $T\ge 0$, which implies that $N(T)\le C_0$. 
\\ \\
As we remarked at the beginning of this subsection, we have to estimate the term $(R)$ more carefully to get a good estimate for $||u(t)||_\infty$. The idea for the estimate of $(R)$ is taken from \cite{gm}. First, we apply a Sobolev inequality: We know that $H^k(\R^d)$ embeds continuously into $L^\infty(\R^d)$ if $k>\frac{d}{2}$. Thus, 
\begin{align*}
(R)&=\int_{t_0}^t\dr s ||e^{-\ri (t-s)H}(w*|u(s)|^2)u(s)||_\infty\\
&\le C\int_{t_0}^t\dr s ||e^{-\ri (t-s)H}(w*|u(s)|^2)u(s)||_{H^k}\, .
\end{align*}
Recall that for any $f\in H^k(\R^d)$, we have
\begin{equation}
||e^{-\ri t(-\Delta)}f||_{H^k}=||f||_{H^k}\, .
\end{equation}
More generally, Lemma \ref{le:W^(k,p) dispersive estimate} with $p=2$ shows that if $k$ is even and $V\in W^{k,\infty}(\R^d)$, then there exists a constant $C>0$ such that
\begin{equation}\label{eq:Sobolev dispersive estimate for p=2 introduction}
||e^{-\ri tH}f||_{H^k}\le C||f||_{H^k}\, .
\end{equation}
Using this inequality, we get
\begin{equation}
(1+|t|)^{\frac{d}{2}}(R)\le C(1+|t|)^{\frac{d}{2}}\int_{t_0}^t\dr s ||(w*|u(s)|^2)u(s)||_{H^k}\, , 
\end{equation}
which can then be estimated by a constant times $M(T)^3$. Combining this with (\ref{eq:Estimate for (1+t)^(d/2) infinity norm of u with M(T) and (R)}), we obtain
\begin{equation}
\sup_{0\le t\le T}(1+|t|)^{\frac{d}{2}}||u(t)||_\infty \le C||e^{\ri H}u_0||_1+ CM(T)^3\, .
\end{equation}
When we combine this estimate with a corresponding estimate for $\sup_{0\le t\le T}||D^k u(t)||_2$, we will obtain an inequality of the form 
\begin{equation}\label{eq:M(T)<=epsilon+CM(T)^3 introduction}
M(T)\le \epsilon+CM(T)^3\, , 
\end{equation}
where $C>0$ is a fixed constant and $\epsilon$ is small if the initial data is small in the sense of \eqref{eq:initial data are small d>=3 small} in the assumptions in Theorem \ref{th:dispersive d>=3 small}. An inequality such as \eqref{eq:M(T)<=epsilon+CM(T)^3 introduction} was the key estimate in \cite{katopusateri}, where the quantity corresponding to our $M(T)$ was called $||u||_{X_T}$. If $M(T)<\infty$, equation \eqref{eq:M(T)<=epsilon+CM(T)^3 introduction} can be re-written as
\begin{equation}\label{eq:epsilon+CM(T)^3-M(T)>=0}
\epsilon+CM(T)^3-M(T)\ge 0\, .
\end{equation}
For $C>0$ fixed and $\epsilon>0$ small enough (depending on $C$), the function 
\begin{equation}
f:[0,\infty)\to\R,\ f(x):=\epsilon+Cx^3-x
\end{equation}
\sloppy is such that $\{f\ge 0\}$ consists of two disjoint intervals that have a strictly positive distance from each other. We call these intervals $I_1$ and $I_2$ and choose them such that $0\in I_1$. Note that $I_1$ is bounded, see also Lemma \ref{le:f>=0 two intervals} and Figure \ref{fi:graph of f}. Here, $x$ plays the role of $M(T)$.
\begin{figure}[h]
\centering
\includegraphics[width=9cm]{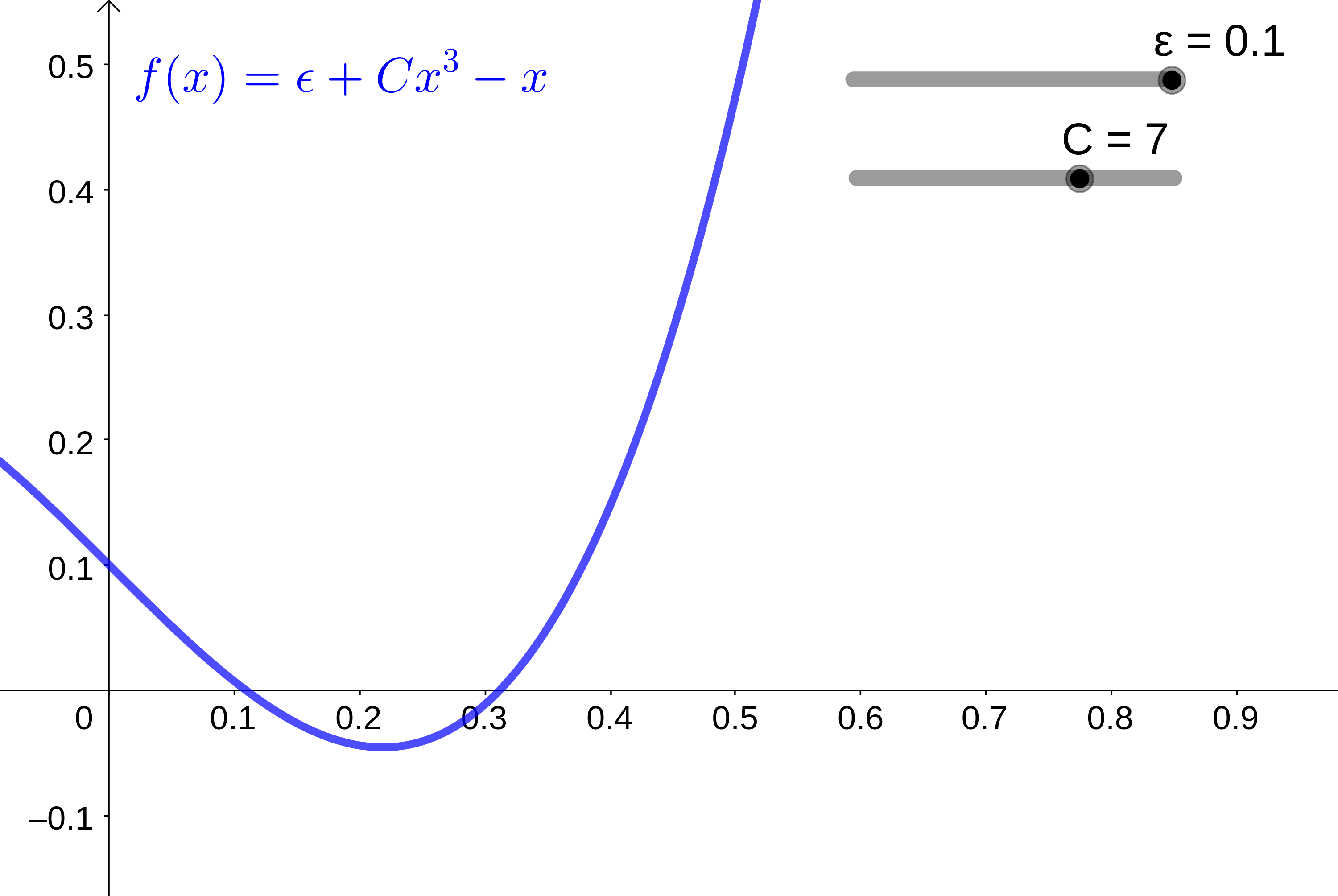}
\caption{This graph shows the function $f:[0,\infty), f(x):=\epsilon+Cx^3-x$ for $\epsilon=0.1$ and $C=7$. Note that the set $\{f\ge 0\}$ consists of two disjoint closed intervals.}\label{fi:graph of f}
\end{figure}
\\ \\
We know that $u$ is a global $H^2$-solution by Theorem \ref{th:global H^2 solutions}. By Theorem \ref{th:H^m solutions for m>d/2}, using the uniqueness of solutions, there exists a $T_{\max}\in (0,\infty]$ such that $u\in C\left([0,T_{\max}), H^k(\R^d)\right)$. Furthermore, by Theorem \ref{th:H^m solutions for m>d/2}, the blow-up alternative holds: If $T_{\max}<\infty$, then $\lim_{t\uparrow T_{\max}}||u||_{H^k}=\infty$ and $\lim_{t\uparrow T_{\max}}||u||_\infty=\infty$. By the Sobolev embedding theorem, the function $[0,\infty)\to[0,\infty],\ T\mapsto M(T)$ is continuous on $[0,T_{\max})$; in particular, $M(T)$ is finite on that interval. 
\\ \\
Assume that $M(0)\le\sup I_1=:C_0<\infty$. We claim that $T_{\max}=\infty$ and $M(T)\le C_0$ for all $T\ge 0$. We have $M(T)\le C_0$ for all $T\in[0,T_{\max})$ by equation (\ref{eq:epsilon+CM(T)^3-M(T)>=0}), Lemma \ref{le:f>=0 two intervals} and the continuity of $T\mapsto M(T)$ on $[0,T_{\max})$. If $T_{\max}<\infty$, then, by the blow-up alternative and the definition of $M(T)$, we get $\lim_{T\uparrow T_{\max}}M(T)=\infty$, which is a contradiction. Thus, $T_{\max}=\infty$ and hence, $M(T)\le C_0$ for all $T\ge 0$. 

\subsection{Organisation}
In Section \ref{s:Tools}, we recall and prove technical results, which we will need for the proof of our main result. We recall various results on solutions to the Hartree type equation from \cite{cazenave}. In order to be able to deal with non-zero external potentials $V$, we look at dispersive estimates for $e^{-\ri t H}$. We recall conditions under which there is an $L^1$--$L^\infty$ dispersive estimate and we prove \eqref{eq:Sobolev dispersive estimate for p=2 introduction}. Section \ref{s:various estimates} is devoted to proving the estimates, which we need for the proof of the main result. In Section \ref{s:conclusion of the main theorem}, we prove Theorem \ref{th:dispersive d>=3 small}. We follow the proof strategy explained in Subsection \ref{ss:proof strategy introduction}. Furthermore, we show an extension of Theorem \ref{th:dispersive d>=3 small} for large data under certain additional assumptions and we prove Theorem \ref{th:cubic NLS main result}.

\subsection{Notations}
We use the convention that all constants with an upper index are greater than or equal to $1$ unless we define them otherwise. For instance, we have
\begin{equation}
C^S, C^{ES}, C^V, C^{DS}, C^{KP}\ge 1\, . 
\end{equation}
For external potentials $V:\R^d\to\R$, define the operator $H:=-\Delta+V$. 
\\ \\
{\bf Acknowledgements.} 
The author would like to express her deepest gratitude to Phan Thành Nam for his continued support and very helpful discussions. The author acknowledges the support from the Deutsche Forschungsgemeinschaft (DFG project Nr. 426365943).
\section{Preliminaries}\label{s:Tools}
In this section, we recall some known results and we prove several technical lemmata, which we will need for our proofs. 
\subsection{Dispersive estimates for $e^{-\ri tH}$}
A natural question is to ask under which conditions on the external potential $V$ the operator $e^{-\ri tH}$ satisfies a dispersive estimate similar to $e^{-\ri t(-\Delta)}$. The proof of the following theorem under condition $(1)$ was provided in \cite[Theorem 1.1]{jounesoffersogge} and under condition $(2)$, a proof can be found in \cite[Theorem 1.1]{rs}.
\begin{theorem}[Dispersive estimate for $e^{-\ri tH}$ in $d\ge 3$]\label{th:dispersive estimate for e^(-itH)}
Let $d\ge 3$ and let $V:\R^d\to\R$. Furthermore, assume that one of the following assumptions is satisfied:
\begin{enumerate}
\item[(1)] 
\begin{enumerate}
\item \sloppy There exist $\eta>0$ and $\alpha>d+4$ such that the multiplication operator $\left(1+|x|^2\right)^{\frac{\alpha}{2}} V(x)$ is a bounded operator from $H^\eta(\R^d)$ to $H^\eta(\R^d)$. 
\item $\widehat{V}\in L^1(\R^d)$.
\item The operator $H$ has purely absolutely continuous spectrum.
\item $0$ is neither an eigenvalue nor a resonance for $H$. That is, there exists no function $\psi\neq 0$ in the weighted $L^2$-space $L^2(\langle x \rangle^\sigma\mathrm d x)$ for some $\sigma\ge 0$ such that $H\psi=0$ in the distributional sense. 
\end{enumerate} 
\item[(2)] $d=3$, 
\begin{equation}
\int_{\R^3}\dr{x}\int_{\R^3}\dr{y} \frac{|V(x)||V(y)|}{|x-y|^2}<(4\pi)^2
\end{equation}
and 
\begin{equation}
\sup_{x\in\R^3} \int_{\R^3}\dr{y} \frac{|V(y)|}{|x-y|}<4\pi\, .
\end{equation}
\end{enumerate}
Then there exists a constant $C^V=C^V(d,V)\ge 1$ such that for all $f\in L^1(\R^d)\cap L^2(\R^d)$ and all $t\in\R\setminus\{0\}$, we have 
\begin{equation}\label{eq:e^(-itH) dispersive estimate for infinity norm}
||e^{-\ri tH}f||_\infty\le C^V\frac{1}{|t|^{\frac{d}{2}}}||f||_1 \, . 
\end{equation}
More generally, if $p\in [2,\infty]$ and $\frac{1}{p}+\frac{1}{p'}=1$, the following \emph{$L^p$-dispersive estimate} holds true: For all $f\in L^{p'}(\R^d)\cap L^2(\R^d)$ and all $t\in\R\setminus\{0\}$, we have
\begin{equation}
||e^{-\ri tH}f||_p\le C^V|t|^{-\frac{d}{2}\left(\frac{1}{p'}-\frac{1}{p}\right)}||f||_{p'} \, . 
\end{equation}
\end{theorem}
\begin{remark}
If condition $(1)$ or $(2)$ is satisfied, then $H$ is a self-adjoint operator. In particular, $e^{-\ri t H}$ is unitary. The $L^p$-dispersive estimate follows from the conservation of the $L^2$ norm, the dispersive estimate (\ref{eq:e^(-itH) dispersive estimate for infinity norm}) and the Riesz-Thorin interpolation theorem, see for example \cite[Theorem 2.1, Lemma 4.1]{nd}. Note that we chose $C^V\ge 1$ and thus, the constant does not change for $p\in[2,\infty]$. 
\end{remark}
\begin{lemma}[$W^{k,p}(\R^d)$-dispersive estimate for $e^{\ri tH}$]\label{le:W^(k,p) dispersive estimate}
Let $d\ge 1$, $2\le p<\infty$ and assume that $V:\R^d\to\R$ is such that $e^{-\ri tH}$ is a unitary operator, which satisfies the dispersive estimate $||e^{-\ri tH}f||_\infty\le C^V |t|^{-\frac{d}{2}}||f||_1$ for some constant $C^V\ge 1$. Let $k\in\N_0$ be even and assume that $V\in W^{k,\infty}(\R^d)$. Then there exists a constant $C^{DS}=C^{DS}(d,k,||V||_{W^{k,\infty}}, C^V)\ge 1$ such that
\begin{equation}
||e^{-\ri tH}f||_{W^{k,p}}\le C^{DS} |t|^{-\frac{d}{2}\left(\frac{1}{p'}-\frac{1}{p}\right)}||f||_{W^{k,p'}}
\end{equation}
for all $f\in W^{k,p'}(\R^d)\cap L^2(\R^d)$. 
\end{lemma}
\begin{remark}\label{re:W^(k,p) dispersive estimate Yajima}
\begin{itemize}
\item[(i)] Yajima proved decay estimates of this form under various decay assumptions on $V$ and its derivatives, see \cite[Theorem 1.3]{yajima} and \cite[Theorem 1.1, Theorem 1.2]{yajima3}. 
\item[(ii)] Lemma \ref{le:W^(k,p) dispersive estimate} provides a simple and certainly not optimal condition, under which an $L^1$--$L^\infty$ decay estimate extends to a $W^{k,p'}$--$W^{k,p}$ decay estimate, namely if $V\in W^{k,\infty}(\R^d)$. This is sufficient for the proof of our main theorem. By contrast, the results by Yajima include the proof of an $L^1$--$L^\infty$ decay estimate, which is just a special case of the $W^{k,p'}$--$W^{k,p}$ decay estimate, see \cite{yajima,yajima3}. Moreover, note that Yajima proved these decay estimates for all $2\le p\le \infty$, whereas we had to exclude $p=\infty$ in Lemma \ref{le:W^(k,p) dispersive estimate}.
\item[(iii)] If we are only interested in the case $p=2$, we do not need the assumption that $e^{-\ri tH}$ satisfies the dispersive estimate $||e^{-\ri tH}f||_\infty\le C^V |t|^{-\frac{d}{2}}||f||_1$. Instead, it suffices to assume that $H$ is a self-adjoint operator and hence, $e^{-\ri tH}$ is unitary. 
\end{itemize}
\end{remark}
\begin{proof}[Proof of Lemma \ref{le:W^(k,p) dispersive estimate} for $p=2$]
The two main ingredients of the proof are the fact that $H$ and $e^{-\ri t H}$ commute and that there exists a constant $C=C(d,k, ||V||_{W^{k,\infty}})>0$ such that
\begin{equation}\label{eq:equivalent Sobolev norm with V}
\frac{1}{C}\sum_{j=0}^\frac{k}{2} ||(-\Delta+V)^j\phi||_2\le ||\phi||_{H^k}\le C\sum_{j=0}^\frac{k}{2} ||(-\Delta+V)^j\phi||_2
\end{equation}
for all $\phi\in H^k(\R^d)$. Suppose \eqref{eq:equivalent Sobolev norm with V} is true. Using that $e^{-\ri t H}$ is unitary and that $H$ and $e^{-\ri t H}$ commute, we get for every $f\in H^k(\R^d)$
\begin{align*}
||e^{-\ri tH}f||_{H^k}&\le C\sum_{j=0}^\frac{k}{2} ||(-\Delta+V)^je^{-\ri tH}f||_2=C\sum_{j=0}^\frac{k}{2} ||H^je^{-\ri tH}f||_2\\
&=C\sum_{j=0}^\frac{k}{2} ||e^{-\ri tH}H^jf||_2\le C\sum_{j=0}^\frac{k}{2} ||H^jf||_2\le C^2 ||f||_{H^k}\, . 
\end{align*}
It remains to show \eqref{eq:equivalent Sobolev norm with V}. 
\subsubsection*{Lower bound}
Let $\ell\in\N$, $m\in2\N_0$ with $2\ell +m\le k$. By the Leibniz rule, we get
\begin{align*}
&\qquad ||(-\Delta+V)^\ell\phi||_{H^m} = ||(-\Delta+V)(-\Delta+V)^{\ell-1}\phi||_{H^m}\\
&\le ||-\Delta(-\Delta+V)^{\ell-1}\phi||_{H^m}+||V(-\Delta+V)^{\ell-1}\phi||_{H^m}\\
&\lesssim ||(-\Delta+V)^{\ell-1}\phi||_{H^{m+2}}+||V||_{W^{m,\infty}}||(-\Delta+V)^{\ell-1}\phi||_{H^m}\\
&\lesssim ||(-\Delta+V)^{\ell-1}\phi||_{H^{m+2}}\, . 
\end{align*}
By  iterating this process, we obtain
\begin{equation}
\sum_{j=0}^\frac{k}{2} ||(-\Delta+V)^j\phi||_2\lesssim ||\phi||_{H^k}\, .
\end{equation}
\subsubsection*{Upper bound}
Let $\ell\in\N_0$, $m\in2\N$ with $2\ell+m\le k$. Again, using the Leibniz rule, we obtain
\begin{align*}
&\qquad ||(-\Delta+V)^\ell\phi||_{H^m}\lesssim ||(-\Delta+V)^\ell\phi||_2+||-\Delta(-\Delta+V)^\ell\phi||_{H^{m-2}}\\
&\le ||(-\Delta+V)^\ell\phi||_2+ ||(-\Delta+V)(-\Delta+V)^\ell\phi||_{H^{m-2}}+||V(-\Delta+V)^\ell\phi||_{H^{m-2}}\\
&\lesssim ||(-\Delta+V)^\ell\phi||_2+ ||(-\Delta+V)^{\ell+1}\phi||_{H^{m-2}}+||V||_{W^{m-2,\infty}}||(-\Delta+V)^\ell\phi||_{H^{m-2}}\\
&\lesssim ||(-\Delta+V)^\ell\phi||_{H^{m-2}}+||(-\Delta+V)^{\ell+1}\phi||_{H^{m-2}}\, .
\end{align*}
Note that $2\ell+(m-2)\le k$ and $2(\ell+1)+(m-2)\le k$. Therefore, we can iterate this estimate and we get
\begin{equation}
||\phi||_{H^k}\lesssim \sum_{j=0}^\frac{k}{2} ||(-\Delta+V)^j\phi||_2 \, .
\end{equation}
\end{proof}
\begin{remark}
The proof of Lemma \ref{le:W^(k,p) dispersive estimate} for $p\in (2,\infty)$ follows the same strategy. Instead of using the unitarity of $e^{-\ri t H}$, this proof uses the $L^p$-dispersive estimate. Another key ingredient of the proof is the following: For $1<p<\infty$ and $k\in\N_0$ even, there exists a constant $C^{ES}=C^{ES}(d,k,p)\ge 1$ such that
\begin{equation}
\frac{1}{C^{ES}}\{||f||_p+||D^k f||_p\}\le ||f||_{W^{k,p}(\R^d)}\le C^{ES}\{||f||_p+||D^k f||_p\}
\end{equation}
for all $f\in W^{k,p}(\R^d)$, where $D^k:=(-\Delta)^{\frac{k}{2}}$. This inequality can be proved using the Gagliardo-Nirenberg inequality, see \cite[Equation (3.14)]{nd}, and the estimate
\begin{equation}
\left|\left|\frac{\partial^2 f}{\partial x_j\partial x_k}\right|\right|_p\le C||\Delta f||_p
\end{equation}
for some constant $C=C(p)>0$ for all $1<p<\infty$, $j,k\in\{1,\cdots,d\}$ and all $f\in C_c^2(\R^d)$, see \cite[Proposition 3, p. 59]{stein}.
\end{remark}
\subsection{The Hartree type equation}\label{ss:hartree type equation}
In this subsection, we consider the Hartree type equation
\begin{equation}\label{eq:Hartree}
\begin{cases}
\ri\partial_t u & =(-\Delta +V)u+(w*|u|^2)u\\
u(0) & =u_0\, . 
\end{cases}
\end{equation}
We collect some results from the literature on existence, uniqueness and continuity of solutions to this equation, see the book by Cazenave \cite{cazenave}. 
\\ \\
Let us first recall the definition of weak and strong solutions to \eqref{eq:Hartree}, see \cite[Definition 3.1.1]{cazenave}.
\begin{definition}[Weak solutions and strong solutions]\label{de:weak and strong solutions}
Let $u_0\in H^1(\R^d)$ and let $0\in I\subset \R$ be an interval. 
\begin{itemize}
\item[(i)] $u$ is called a \emph{weak $H^1$-solution} to (\ref{eq:Hartree}) if $u(0)=u_0$,
\begin{equation}
u\in L^\infty\left(I, H^1(\R^d)\right)\cap W^{1,\infty}\left(I,H^{-1}(\R^d)\right)
\end{equation}
and 
\begin{equation}
0=-\ri\partial_t u +(-\Delta +V)u+(w*|u|^2)u\ \mathrm{in\ } H^{-1}(\R^d) \mathrm{\ for \ almost \ every\ } t\in I.
\end{equation}
\item[(ii)] $u$ is called a \emph{strong $H^1$-solution} to (\ref{eq:Hartree}) if $u(0)=u_0$, 
\begin{equation}
u\in C\left(I, H^1(\R^d)\right)\cap C^1\left(I,H^{-1}(\R^d)\right)
\end{equation}
and
\begin{equation}
0=-\ri\partial_t u +(-\Delta +V)u+(w*|u|^2)u \mathrm{\ in\ } H^{-1}(\R^d) \mathrm{\ for \  every\ } t\in I.
\end{equation}
\end{itemize}
\end{definition}
Let us now define well-posedness for the Hartree type equation \eqref{eq:Hartree}, see \cite[Definition 3.1.5]{cazenave}.
\begin{definition}[Local well-posedness in $H^1$]\label{de:local well-posedness}
We call the initial value problem (\ref{eq:Hartree}) \emph{locally well-posed} in $H^1$ if the following properties hold:
\begin{itemize}
\item[(i)] There is uniqueness in $H^1
$ for (\ref{eq:Hartree}), that is, for every $u_0\in H^1(\R^d)$ and for every interval $0\in I\subset\R$, any two weak solutions to (\ref{eq:Hartree}) on $I$ coincide. 
\item[(ii)] For every $u_0\in H^1(\R^d)$, there exists a strong $H^1$-solution $u$ defined on a maximal interval $(-T_{\min},T_{\max})$, where $T_{\min},T_{\max}\in (0,\infty]$. $T_{\min}$ and $T_{\max}$ can depend on the initial data $u_0$. 
\item[(iii)] There is the blow-up alternative: If $T_{\max}<\infty$, then $\lim_{T\uparrow T_{\max}}||u(t)||_{H^1}=\infty$ (similarly for $T_{\min}$). 
\item[(iv)] The solution $u$ depends continuously on the initial data $u_0$: If $u_0^n\to u_0$ in $H^1(\R^d)$ as $n\to\infty$ and  $I\subset (-T_{\min},T_{\max})$ is a closed and bounded interval, then there exists $N\in\N$ such that for all $n\ge N$, we know that the corresponding solution to the Hartree type equation $u^n$ is defined on $I$. Furthermore, $u^n\to u$ in $C\left(I, H^1(\R^d)\right)$ as $n\to\infty$. 
\end{itemize}
\end{definition}
Let us define the energy, see \cite[Equation (3.3.9)]{cazenave}.
\begin{definition}[Energy]\label{de:energy}
Let $V:\R^d\to\R$, $V\in L^{p_V}(\R^d)+L^{q_V}(\R^d)$ for some $p_V,q_V\ge 1$ with $p_V,q_V>\frac{d}{2}$ and let $w:\R^d\to\R$ be an even function with $w\in L^{p_w}(\R^d)$ for some $p_w\ge 1$ with $p_w>\frac{d}{4}$. For any $v\in H^1(\R^d)$, we define the \emph{energy} by
\begin{equation}
E(v):=\int_{\R^d}|\nabla v|^2+\int_{\R^d}\dr x V(x) |v(x)|^2 +\frac{1}{2}\int_{\R^d}\dr x (w*|v|^2)(x)|v(x)|^2\, .
\end{equation}
\end{definition}
\begin{remark}\label{re:energy finite}
For every $v\in H^1(\R^d)$, we have $E(v)\in\R$ by \cite[Proposition 3.2.2(i), Proposition 3.2.9(i)]{cazenave}. From the proof of these statements, we also know that the energy $E(v)$ can be bounded by a constant $C<\infty$ that only depends on $||v||_{H^1(\R^d)}$, $||V||_{L^{p_V}+L^{q_V}}$ and $||w||_{p_w}$. Furthermore, $C\to 0$ as $||v||_{H^1(\R^d)}\to 0$.
\end{remark}
The Hartree type equation \eqref{eq:Hartree} is locally well-posed and mass and energy are conserved, see \cite[Theorem 4.3.1, Remark 3.3.4]{cazenave}.
\begin{theorem}[Local well-posedness and conservation of mass and energy]\label{th:local well-posedness and conservation of mass and energy}
Assume that $V:\R^d\to\R$, $V\in L^{p_V}(\R^d)+L^{q_V}(\R^d)$ for some $p_V,q_V\ge 1$ with $p_V,q_V>\frac{d}{2}$. Moreover, assume that $w:\R^d\to\R$ is an even function with $w\in L^{p_w}(\R^d)$ for some $p_w\ge 1$ with $p_w>\frac{d}{4}$. Let $u_0\in H^1(\R^d)$. Then the initial value problem (\ref{eq:Hartree}) is locally well-posed in $H^1$. Let $u$ be the corresponding strong $H^1$-solution to (\ref{eq:Hartree}). The mass and the energy are conserved:
\begin{align}
||u(t)||_2&=||u_0||_2\\
E(u(t))&=E(u_0)
\end{align} 
for all $t\in (-T_{\min}, T_{\max})$. 
\end{theorem}
Under suitable assumptions on $V$ and $w$, strong $H^1$-solutions are global, see \cite[Corollary 6.1.2]{cazenave}, and we have $H^2$ regularity, see \cite[Theorem 5.3.1, Remark 5.3.3]{cazenave}.
\begin{theorem}[Global $H^2$-solutions]\label{th:global H^2 solutions}
Let $V:\R^d\to\R$ with $V\in L^{p_V}(\R^d)+L^{q_V}(\R^d)$, where $p_V,q_V\ge 1$ and $p_V,q_V>\frac{d}{2}$. Let $w:\R^d\to\R$ be an even function with $w\in L^{p_w}(\R^d)$ for some $p_w\ge 1$ with $p_w>\frac{d}{4}$. Moreover, assume that $w_-\in L^\infty(\R^d)+L^{q_w}(\R^d)$ for some $q_w$ with $q_w\ge\max\left\{1,\frac{d}{2}\right\}$ if $d\neq 2$ and $q_w>1$ if $d=2$. Let $u_0\in H^2(\R^d)$ and let $u$ be the unique global strong $H^1$-solution from Theorem \ref{th:local well-posedness and conservation of mass and energy}. Then $u$ is a global solution and $u\in C\left(\R, H^2(\R^d)\right)$. Furthermore, 
\begin{equation}
\sup_{t\in\R} ||\nabla u(t)||_2<\infty\, .
\end{equation}
\end{theorem}
Duhamel's formula will be one of the key ingredients of our proof.
\begin{lemma}[Duhamel's formula]\label{le:Duhamel}
Let $w\in L^{p_w}(\R^d)$ for some $p_w\ge 1$ with $p_w>\frac{d}{4}$ be an even real-valued function. Let $V:\R^d\to\R$ belong to $L^{p_V}(\R^d)+L^{q_V}(\R^d)$ for $\max\left\{\frac{d}{2},2\right\}<p_V,q_V\le\infty$. Let $u_0\in H^1(\R^d)$ and let $u$ be the unique strong solution to (\ref{eq:Hartree}) given by Theorem \ref{th:local well-posedness and conservation of mass and energy} defined on its maximal time interval of existence $I$ with $0\in I\subset\R$. Then
\begin{equation}
u(t) = e^{-\ri tH}u_0-\ri\int_0^t\dr s e^{-\ri (t-s)H}(w*|u(s)|^2)u(s)
\end{equation}
for every $t\in I$.
\end{lemma}
\begin{proof}
By \cite[Theorem on p. 25]{fa2nam}, the operator $H=-\Delta+V$ is self-adjoint on $L^2(\R^d)$ with domain $\mathcal D(H)=H^2(\R^d)$. By \cite[Proposition 3.2.9(i)]{cazenave}, the map $u\mapsto (w*|u|^2)u$ belongs to $C\left(H^1(\R^d),H^{-1}(\R^d)\right)$. Since $u$ is a strong $H^1$-solution, we have $u\in C\left(I, H^1(\R^d)\right)\cap C^1\left(I,H^{-1}(\R^d)\right)$ by Definition \ref{de:weak and strong solutions}. Combining these two facts, we find that $(w*|u|^2)u\in C\left(I, H^{-1}(\R^d)\right)\subset C\left(I, H^{-2}(\R^d)\right)=C\left(I, (\mathcal D(H))^*\right)$. We can now conclude by \cite[Remark 1.6.1(ii)]{cazenave}.
\end{proof}
\begin{remark}[Duhamel's formula for initial data at $t=1$]\label{re:Duhamel starting at t=1}
If we consider the Hartree type equation with initial data at $t=1$
\begin{equation}
\begin{cases}
\ri\partial_t u & =(-\Delta +V)u+(w*|u|^2)u\\
u(1) & =u_1\, , 
\end{cases}
\end{equation}
the corresponding Duhamel's formula is 
\begin{equation}
u(t) = e^{-\ri (t-1)H}u_1-\ri\int_1^t\dr s e^{-\ri (t-s)H}(w*|u(s)|^2)u(s)\, . 
\end{equation}
\end{remark}
\begin{remark}[Generalised Duhamel's formula]\label{re:generalised Duhamel}
The nonlinearity $(w*|u|^2)u$ in Lemma \ref{le:Duhamel} can be replaced by a more general nonlinearity $f\in C\left(I,H^{-2}(\R^d)\right)$, see \cite[Remark 1.6.1(ii)]{cazenave}.
\end{remark}
An analogous result to the following Lemma was proved in \cite[Lemma 4.10.2]{cazenave} for local nonlinearities $g(u)$, where $g$ is a function $g:\C\to\C$. We prove it for the non-local interaction $g(u)=(w*|u|^2)u$ and for $g(u)=Vu$. We will need this result for the proof of Theorem \ref{th:H^m solutions for m>d/2}.
\begin{lemma}\label{le:H^m boundedness and continuity of interaction part}
Let $d\ge 1$ and let $k\in\N$ with $k>\frac{d}{2}$. Moreover, let $w\in L^1(\R^d)$. For $u:\R^d\to\C$, let either $g(u)=Vu$ for $V\in W^{k,\infty}(\R^d)$ or $g(u):=(w*|u|^2)u$. Let $M>0$. Then there exists a constant $C(M)>0$ such that the following properties hold:
\begin{itemize}
\item[(i)] For all $u\in H^k(\R^d)$ with $||u||_\infty\le M$, we have
\begin{equation}
||g(u)||_{H^k}\le C(M)||u||_{H^k}\, . 
\end{equation}
\item[(ii)] For all $u,v\in H^k(\R^d)$ with $||u||_\infty,||v||_\infty\le M$, we have
\begin{equation}
||g(u)-g(v)||_2\le C(M)||u-v||_2\, . 
\end{equation}
\item[(iii)] For all $u,v\in H^k(\R^d)$ with $||u||_{H^k(\R^d)},||v||_{H^k(\R^d)}\le M$, we have
\begin{equation}
||g(u)-g(v)||_{H^k}\le C(M)||u-v||_{H^k}\, . 
\end{equation}
\end{itemize}
\end{lemma}

\begin{proof}
{\bf The case $g(u)=Vu$.}
Let $u,v\in H^k(\R^d)$. By the Leibniz rule, we have
\begin{equation}
||g(u)||_{H^k}\le C ||V||_{W^{k,\infty}}||u||_{H^k}
\end{equation}
for a constant $C>0$ only depending on $k$ and the dimension $d$. This shows (i). Similarly, we have
\begin{equation}
||g(u)-g(v)||_{H^k}=||V(u-v)||_{H^k}\le C ||V||_{W^{k,\infty}}||u-v||_{H^k}\, ,
\end{equation}
so (iii) holds. Moreover, we have
\begin{equation}
||g(u)-g(v)||_2=||V(u-v)||_2 \le ||V||_\infty||u-v||_2\, ,
\end{equation} 
which shows (ii). 
\\ \\
{\bf The case $g(u):=(w*|u|^2)u$.}
Let $M>0$. 
\subsubsection*{Proof of (i)}
Let $u\in H^k(\R^d)$ with $||u||_\infty\le M$. We have
\begin{equation}
||g(u)||_{H^k}\le C^{ES}\left(||g(u)||_2+||D^k g(u)||_2\right)
\end{equation}
for some constant $C^{ES}\ge 1$. 
\\ \\
For $||g(u)||_2$, we estimate
\begin{align*}
||g(u)||_2&= ||(w*|u|^2)u||_2\le ||w*|u|^2||_\infty||u||_2\le ||w||_1||u||_\infty^2||u||_2\\
&\le ||w||_1M^2||u||_2\le ||w||_1M^2||u||_{H^k}\, .
\end{align*}
For $||D^k g(u)||_2$, we use the Kato-Ponce inequality, which states that
\begin{equation}
||D^k(fh)||_2\le C^{KP}\left(||D^k f||_2||h||_\infty+||f||_\infty||D^k h||_2\right)
\end{equation}
for all $f,h$ and for some fixed constant $C^{KP}\ge 1$, see \cite[Theorem 1.4(2)]{gk}. We get
\begin{align*}
&\quad ||D^k g(u)||_2=||D^k [(w*|u|^2)u]||_2\\
&\le C^{KP}(||D^k (w*|u|^2)||_2||u||_\infty+||w*|u|^2||_\infty||D^k u||_2)\\
&\le C^{KP}(||w*(D^k(u\overline{u}))||_2||u||_\infty+||w||_1||u||_\infty^2||D^k u||_2)\\
&\le C^{KP}(||w||_1||(D^k(u\overline{u}))||_2M+||w||_1M^2||u||_{H^k})\\
&\le C^{KP}||w||_1 M(2C^{KP}||D^ku||_2||u||_\infty+M||u||_{H^k})\\
&\le 3(C^{KP})^2||w||_1 M^2||u||_{H^k}\, .
\end{align*}
To sum up,
\begin{equation}
||g(u)||_{H^k(\R^d)}\le C^{ES}\left(||g(u)||_2+||D^k g(u)||_2\right)\le 4C^{ES}(C^{KP})^2||w||_1 M^2||u||_{H^k}\, . 
\end{equation}
\subsubsection*{Proof of (ii)}
Let $u,v\in H^k(\R^d)$ with $||u||_\infty,||v||_\infty\le M$. We have
\begin{align*}
||g(u)-g(v)||_2&=||(w*|u|^2)u-(w*|v|^2)v||_2\\
&\le ||(w*|u|^2)(u-v)||_2+||(w*(|u|^2-|v|^2))v||_2\, .
\end{align*}
We estimate the first term by
\begin{align*}
||(w*|u|^2)(u-v)||_2&\le ||w*|u|^2||_\infty||u-v||_2\le ||w||_1||u||_\infty^2||u-v||_2\\
&\le ||w||_1M^2||u-v||_2
\end{align*}
and the second term by
\begin{align*}
&\qquad ||(w*(|u|^2-|v|^2))v||_2\le||w||_1||(|u|+|v|)(|u|-|v|)||_2||v||_\infty\\
&\le||w||_1M(||u||_\infty+||v||_\infty)|||u|-|v|||_2\le 2||w||_1M^2||u-v||_2\, . 
\end{align*}
Hence,
\begin{equation}
||g(u)-g(v)||_2\le ||(w*|u|^2)(u-v)||_2+||(w*(|u|^2-|v|^2))v||_2\le 3||w||_1M^2||u-v||_2\, . 
\end{equation}
\subsubsection*{Proof of (iii)}
Let $u,v\in H^k(\R^d)$ with $||u||_{H^k(\R^d)},||v||_{H^k(\R^d)}\le M$. We have
\begin{equation}
||g(u)-g(v)||_{H^k}\le C^{ES}\left(||g(u)-g(v)||_2+||D^k [g(u)-g(v)]||_2\right)
\end{equation}
for some constant $C^{ES}\ge 1$. We estimated the first summand in the proof of (ii). Thus, it remains to estimate the second summand. We have 
\begin{align*}
&\quad||D^k [g(u)-g(v)]||_2= ||D^k [(w*|u|^2)u-(w*|v|^2)v]||_2\\
&\le ||D^k [(w*|u|^2)(u-v)]||_2+||D^k [(w*(|u|^2-|v|^2))v]||_2\, . 
\end{align*}
By the Kato-Ponce inequality and the Sobolev embedding theorem ($||f||_\infty\le C^S||f||_{H^k}$ for some $C^S\ge 1$), we get
\begin{align*}
&\qquad||D^k [(w*|u|^2)(u-v)]||_2\\
&\le C^{KP}\left(||D^k [w*|u|^2]||_2||u-v||_\infty+||w*|u|^2||_\infty||D^k [u-v]||_2\right)\\
&\le C^{KP}\left(||w||_1||D^k [u\overline u]||_2||u-v||_\infty+||w||_1||u||_\infty^2||u-v||_{H^k}\right)\\
&\le C^{KP}||w||_1\left(C^{KP}(||D^k u||_2||\overline u||_\infty+||u||_\infty||D^k\overline u||_2)||u-v||_\infty+||u||_\infty^2||u-v||_{H^k}\right)\\
&\le 3(C^S)^2 (C^{KP})^2||w||_1 M^2||u-v||_{H^k}\, . 
\end{align*}
Similarly, we obtain
\begin{align*}
& \qquad ||D^k [(w*(|u|^2-|v|^2))v]||_2\\
&\le C^{KP}\left(||D^k [w*(|u|^2-|v|^2)]||_2||v||_\infty+||w*(|u|^2-|v|^2)||_\infty||D^k v||_2\right)\\
&\le C^{KP}||w||_1\left(||D^k [(u-v)\overline u+v\overline{(u-v)}||_2 ||v||_\infty+||(|u|+|v|)(|u|-|v|)||_\infty ||v||_{H^k}\right)\\
&\le C^{KP}||w||_1\Big(C^{KP}\big(||D^k[u-v]||_2||u||_\infty+||u-v||_\infty ||D^ku||_2+||D^kv||_2||u-v||_\infty\\
&\qquad +||v||_\infty||D^k[u-v]||_2\big)||v||_\infty+(||u||_\infty +||v||_\infty)||u-v||_\infty||v||_{H^k}\Big)\\
&\le 6 (C^S)^2 (C^{KP})^2 ||w||_1 M^2 ||u-v||_{H^k}\, . 
\end{align*}
Combining these two estimates, we get
\begin{equation}
||D^k [g(u)-g(v)]||_2\le 9(C^S)^2 (C^{KP})^2 ||w||_1 M^2 ||u-v||_{H^k}\, .
\end{equation}
\end{proof}
Finally, let us look at the existence of $H^k$-solutions for $k>\frac{d}{2}$ and the corresponding blow-up criterion; compare with \cite[Theorem 4.10.1]{cazenave}.
\begin{theorem}[$H^k$-solutions for $k>\frac{d}{2}$]\label{th:H^m solutions for m>d/2}
Let $d\ge 1$ and $k>\frac{d}{2}$. Let $V:\R^d\to\R$ with $V\in L^{p_V}(\R^d)+L^{q_V}(\R^d)$, where $ p_V,q_V\ge 1$ and $p_V,q_V>\frac{d}{2}$. Let $w:\R^d\to\R$ be an even function with $w\in L^{p_w}(\R^d)$ for some $p_w\ge 1$ with $p_w>\frac{d}{4}$. Moreover, assume that $w_-\in L^\infty(\R^d)+L^{q_w}(\R^d)$ for some $q_w$ with $q_w\ge\max\left\{1,\frac{d}{2}\right\}$ if $d\neq 2$ and $q_w>1$ if $d=2$. Let $u_0\in H^k(\R^d)$. Then there exist $T_{\min},T_{\max}\in (0,\infty]$ and a unique maximal strong solution $u\in C\left((-T_{\min},T_{\max}), H^k(\R^d)\right)$ of (\ref{eq:Hartree}). Moreover, the blow-up alternative holds: If $T_{\max}<\infty$, then $\lim_{t\uparrow T_{\max}}||u||_{H^k}=\infty$ and $\lim_{t\uparrow T_{\max}}||u||_\infty=\infty$ (similarly for $T_{\min}$).
\end{theorem}
\begin{proof}
The proof of this result is provided in step 1 and step 2 of the proof of  \cite[Theorem 4.10.1]{cazenave} up to a small modification. In step 1, \cite[Lemma 4.10.2]{cazenave} is used, which was only proved for local nonlinearities. In our case, we use Lemma \ref{le:H^m boundedness and continuity of interaction part} for both the interaction part $g(u)=(w*|u|^2)u$ and for the part with the external potential $g(u)=Vu$. In step 2, Cazenave uses the uniqueness of the solution from  \cite[Proposition 4.2.9]{cazenave}. Note that alternatively, we can use the uniqueness in $H^1$, which follows from Theorem \ref{th:local well-posedness and conservation of mass and energy}. 
\end{proof}
\begin{remark}\label{re:maximal existence times for H^1 and H^k solutions}
Note that the maximal existence times $T_{\min}, T_{\max}$ in Theorem \ref{th:H^m solutions for m>d/2} do not necessarily agree with those of $H^1$-solutions from Theorem \ref{th:local well-posedness and conservation of mass and energy}. In particular, a solution could be an $H^k$-solution defined on a bounded maximal time interval $I$ but it could be possible to extend the solution to $\R$ as a global $H^1$-solution. 
\end{remark}
\section{Various estimates}\label{s:various estimates}
For notational reasons, it is more convenient to work with the \emph{Hartree type equation with initial data at $t=1$} 
\begin{equation} 
\begin{cases}
\ri\partial_t u & =-\Delta u+ Vu+(w*|u|^2)u\\
u(1) & =u_1
\end{cases}
\end{equation}
instead. The smallness assumption on the initial data for this equation is 
\begin{equation}
||e^{\ri H}u_1||_1\, , ||u_1||_{H^k}\le\epsilon_0
\end{equation}
and
\begin{equation}
||e^{\ri H}(\partial_t u)(1)||_1\, ,||(\partial_t u)(1)||_{H^k}\le\tilde\epsilon_0\, .
\end{equation}
In the setting of the Hartree type equation with initial data at $t=1$, Theorem \ref{th:dispersive d>=3 small} states that 
\begin{equation}
||u(t)||_\infty\le C_0\frac{1}{|t|^{\frac{d}{2}}}
\end{equation}
and
\begin{equation}
||\partial_t u(t)||_\infty\le \tilde C_0\frac{1}{|t|^{\frac{d}{2}}}
\end{equation}
for all $t\ge 1$. 
\\ \\
This section is devoted to proving the estimates we need in order to prove Theorem \ref{th:dispersive d>=3 small}. Suppose that the assumptions of Theorem \ref{th:dispersive d>=3 small} are satisfied. Let $T\ge 1$ and let $1\le t\le T$. Moreover, assume that $1\le s\le t$. 
\begin{definition}[Definition of $M(T)$]\label{de:M(T) for d>=3 small}
Define for $T\ge 1$
\begin{equation}
M(T):=\sup_{1\le t\le T} |t|^{\frac{d}{2}}||u(t)||_\infty+\sup_{1\le t\le T}||D^ku(t)||_2 +||u_1||_2\, . 
\end{equation}
\end{definition}
Let us start by proving an estimate for $\sup_{1\le t\le T} |t|^{\frac{d}{2}}||u(t)||_\infty$. In particular, we will prove both a direct estimate and a Sobolev type estimate for the term $||e^{-\ri (t-s)H}(w*|u(s)|^2)u(s)||_\infty$. Define $t_0:=\max\{1,t-1\}$. 
\begin{lemma}[Estimate for $||e^{-\ri (t-1)H}u_1||_\infty$]\label{le:estimate for the first part in Duhamel infinity norm in d>=3 small}
We have
\begin{equation}
|t|^{\frac{d}{2}}||e^{-\ri (t-1)H}u_1||_\infty \le C^V||e^{\ri H}u_1||_1\, . 
\end{equation}
\end{lemma}
\begin{proof}
By the dispersive estimate (Theorem \ref{th:dispersive estimate for e^(-itH)}), we get
\begin{equation}
|t|^{\frac{d}{2}}||e^{-\ri (t-1)H}u_1||_\infty=|t|^{\frac{d}{2}}||e^{-\ri tH}e^{\ri H}u_1||_\infty\le C^V||e^{\ri H}u_1||_1\, . 
\end{equation}
\end{proof}
\begin{lemma}[Direct estimate]\label{le:direct estimate d>=3 small}
We have 
\begin{equation}
||e^{-\ri (t-s)H}(w*|u(s)|^2)u(s)||_\infty\le C^V||w||_1M(T)^3|t-s|^{-\frac{d}{2}}|s|^{-\frac{d}{2}}\, .
\end{equation}
\end{lemma}
\begin{proof}
By the dispersive estimate (Theorem \ref{th:dispersive estimate for e^(-itH)}), Hölder's inequality, Young's inequality and the conservation of the $L^2$ norm, see Theorem \ref{th:local well-posedness and conservation of mass and energy}, we have
\begin{align*}
&\qquad ||e^{-\ri (t-s)H}(w*|u(s)|^2)u(s)||_\infty\le C^V|t-s|^{-\frac{d}{2}}||(w*|u(s)|^2)u(s)||_1\\
&\le C^V|t-s|^{-\frac{d}{2}}||(w*|u(s)|^2)||_1||u(s)||_\infty\le C^V|t-s|^{-\frac{d}{2}}||w||_1|||u(s)|^2||_1||u(s)||_\infty\\
&\le C^V||w||_1|t-s|^{-\frac{d}{2}}||u_1||_2^2||u(s)||_\infty\le C^V||w||_1M(T)^3|t-s|^{-\frac{d}{2}}|s|^{-\frac{d}{2}}\, .
\end{align*}
\end{proof}
\begin{lemma}[Sobolev type estimate]\label{le:Sobolev type estimate in d>=3 small}
We have
\begin{equation}
||e^{-\ri(t-s)H}(w*|u(s)|^2)u(s)||_\infty\le C^{SE} ||w||_1M(T)^3|s|^{-d}\, ,
\end{equation}
where $C^{SE}:= 4C^S C^{DS}C^{ES}(C^{KP})^2$. 
\end{lemma}
\begin{proof}
By Sobolev's inequality for $k>\frac{d}{2}$ and Lemma \ref{le:W^(k,p) dispersive estimate} for $p=2$, we have
\begin{align*}
&\quad ||e^{-\ri(t-s)H}(w*|u(s)|^2)u(s)||_\infty\le C^S ||e^{-\ri(t-s)H}(w*|u(s)|^2)u(s)||_{H^k}\\
&\le C^S C^{DS}||(w*|u(s)|^2)u(s)||_{H^k}\\
&\le C^S C^{DS}C^{ES}\left(||(w*|u(s)|^2)u(s)||_2+||D^k[(w*|u(s)|^2)u(s)]||_2\right)
\end{align*}
for some constant $C^{ES}\ge 1$. 
\subsubsection*{Estimate for $||(w*|u(s)|^2)u(s)||_2$.}
We have 
\begin{align*}
&\qquad ||(w*|u(s)|^2)u(s)||_2\le ||(w*|u(s)|^2)||_\infty||u(s)||_2\\
&\le ||w||_1|||u(s)|^2||_\infty||u_1||_2\le ||w||_1||u(s)||_\infty^2||u_1||_2\le ||w||_1M(T)^3|s|^{-d}\, . 
\end{align*}
\subsubsection*{Estimate for $||D^k[(w*|u(s)|^2)u(s)]||_2$.}
We use the Kato-Ponce inequality to get
\begin{align*}
&\quad ||D^k[(w*|u(s)|^2)u(s)]||_2\\
&\le C^{KP}\left(||D^k[w*|u(s)|^2]||_2||u(s)||_\infty + ||w*|u(s)|^2||_\infty||D^ku(s)||_2\ \right)\\
&\le C^{KP}\left(||w*(D^k[|u(s)|^2])||_2||u(s)||_\infty + ||w||_1|||u(s)|^2||_\infty ||D^ku(s)||_2\ \right)\\
&\le C^{KP}\left(||w||_1||D^k[u(s)\overline{u(s)}]||_2||u(s)||_\infty + ||w||_1||u(s)||_\infty^2 ||D^ku(s)||_2\ \right)\\
&\le C^{KP}\Big(||w||_1C^{KP}(||D^ku(s)||_2||\overline{u(s)}||_\infty||D^k\overline{u(s)}||_2||u(s)||_\infty)||u(s)||_\infty\\
&\qquad  +  ||w||_1||u(s)||_\infty^2 ||D^ku(s)||_2\ \Big)\\
&\le 3(C^{KP})^2||w||_1||u(s)||_\infty^2 ||D^ku(s)||_2\\
&\le 3(C^{KP})^2||w||_1M(T)^3|s|^{-d}\, . 
\end{align*}
\subsubsection*{Conclusion.}
We get
\begin{align*}
&\quad ||e^{-\ri(t-s)H}(w*|u(s)|^2)u(s)||_\infty\\
&\le C^S C^{DS}C^{ES}\left(||(w*|u(s)|^2)u(s)||_2+||D^k[(w*|u(s)|^2)u(s)]||_2\right)\\
&\le C^S C^{DS}C^{ES}\left(||w||_1M(T)^3|s|^{-d} + 3(C^{KP})^2 ||w||_1M(T)^3|s|^{-d}\right)\\
&\le 4C^S C^{DS}C^{ES}(C^{KP})^2||w||_1M(T)^3|s|^{-d}\\
&\le C^{SE} ||w||_1M(T)^3|s|^{-d}\, ,
\end{align*}
\sloppy where we set $C^{SE}:= 4C^S C^{DS}C^{ES}(C^{KP})^2$. Note that we used the fact that $C^S,C^{DS},C^{ES},C^{KP}\ge 1$. 
\end{proof} 
\begin{corollary}[Estimate for $|t|^{\frac{d}{2}}||u(t)||_\infty$]\label{co:estimate infinity norm d>=3 small}
We have
\begin{equation}
|t|^{\frac{d}{2}}||u(t)||_\infty\le C^V||e^{\ri H}u_1||_1+C^{\infty E}||w||_1M(T)^3\, ,
\end{equation}
where $C^{\infty E}:=\frac{2^{2+\frac{d}{2}}}{d-2}C^V+ 2^{\frac{d}{2}}C^{SE}$.
\end{corollary}
\begin{proof}
By Duhamel's formula, see Lemma \ref{le:Duhamel}, we have
\begin{align*}
u(t) &= e^{-\ri (t-1)H}u_1-\ri\int_1^{t_0}\dr s e^{-\ri (t-s)H}(w*|u(s)|^2)u(s)\\
&\qquad-\ri\int_{t_0}^t\dr s e^{-\ri (t-s)H}(w*|u(s)|^2)u(s)\, ,
\end{align*}
so we obtain
\begin{align*}
||u(t)||_\infty &=|| e^{-\ri (t-1)H}u_1||_\infty+\int_1^{t_0}\dr s ||e^{-\ri (t-s)H}(w*|u(s)|^2)u(s)||_\infty\\
&\qquad+\int_{t_0}^t\dr s ||e^{-\ri (t-s)H}(w*|u(s)|^2)u(s)||_\infty\, .
\end{align*}
We use Lemma \ref{le:estimate for the first part in Duhamel infinity norm in d>=3 small} for the first term, Lemma \ref{le:direct estimate d>=3 small} for the second term and Lemma \ref{le:Sobolev type estimate in d>=3 small} for the third term to get
\begin{align*}
&\quad |t|^{\frac{d}{2}}||u(t)||_\infty\le C^V||e^{\ri H}u_1||_1+\int_1^{t_0}\dr s C^V||w||_1M(T)^3|t-s|^{-\frac{d}{2}}|s|^{-\frac{d}{2}}|t|^{\frac{d}{2}}\\
&\qquad + \int_{t_0}^t\dr s C^{SE} ||w||_1M(T)^3|s|^{-d}|t|^{\frac{d}{2}}\\
&\le C^V||e^{\ri H}u_1||_1+||w||_1M(T)^3\left(C^V\int_1^{t_0}\dr s |t-s|^{-\frac{d}{2}}|s|^{-\frac{d}{2}}|t|^{\frac{d}{2}} + C^{SE}\int_{t_0}^t\dr s  |s|^{-d}|t|^{\frac{d}{2}}\right)\, .
\end{align*}
We will estimate the integral terms separately, depending on the value of $t$. 
\subsubsection*{Estimate for the integrals for $t\in[1,2]$.}
By definition, $t_0=\max\{1,t-1\}$, so we have $t_0=1$ in this case. Thus, the first integral is equal to zero. For the second integral, we get by $t\le2$ and $s\ge 1$
\begin{equation}
\int_{t_0}^t\dr s  |s|^{-d}|t|^{\frac{d}{2}}\le 2^{\frac{d}{2}}\, .
\end{equation}
\subsubsection*{Estimate for the integrals for $t>2$.}
If $t>2$, then $t_0=t-1>\frac{t}{2}$. Thus, by symmetry, we can write the first integral as
\begin{align*}
&\quad \int_1^{t_0}\dr s |t-s|^{-\frac{d}{2}}|s|^{-\frac{d}{2}}|t|^{\frac{d}{2}}=2\int_1^{\frac{t}{2}}\dr s |t-s|^{-\frac{d}{2}}|s|^{-\frac{d}{2}}|t|^{\frac{d}{2}}= 2\cdot 2^{\frac{d}{2}}\int_1^{\frac{t}{2}}\dr s|s|^{-\frac{d}{2}}\\
&\le 2^{1+\frac{d}{2}}\int_1^\infty\dr s|s|^{-\frac{d}{2}}\le 2^{1+\frac{d}{2}}\frac{-1}{1-\frac{d}{2}}=\frac{2^{2+\frac{d}{2}}}{d-2}\, , 
\end{align*}
where we used that $t-s\ge\frac{t}{2}$. For the second integral, using $s\ge 1$, $s\ge\frac{t}{2}$ and $t-t_0=1$, we get
\begin{equation}
\int_{t_0}^t\dr s  |s|^{-d}|t|^{\frac{d}{2}}\le \int_{t_0}^t\dr s  |s|^{-\frac{d}{2}}|s|^{-\frac{d}{2}}|t|^{\frac{d}{2}}\le 2^{\frac{d}{2}}\, .
\end{equation}
\subsubsection*{Conclusion.}
In both cases, we can estimate
\begin{align*}
&\quad |t|^{\frac{d}{2}}||u(t)||_\infty\\
&\le C^V||e^{\ri H}u_1||_1+||w||_1M(T)^3\left(C^V\int_1^{t_0}\dr s |t-s|^{-\frac{d}{2}}|s|^{-\frac{d}{2}}|t|^{\frac{d}{2}} + C^{SE}\int_{t_0}^t\dr s  |s|^{-d}|t|^{\frac{d}{2}}\right)\\
&\le C^V||e^{\ri H}u_1||_1+||w||_1M(T)^3\left(C^V\cdot \frac{2^{2+\frac{d}{2}}}{d-2}+ C^{SE}\cdot 2^{\frac{d}{2}}\right)\\
&\le C^V||e^{\ri H}u_1||_1+C^{\infty E}||w||_1M(T)^3\, ,
\end{align*}
where $C^{\infty E}:=\frac{2^{2+\frac{d}{2}}}{d-2}C^V+ 2^{\frac{d}{2}}C^{SE}$.
\end{proof}
Next, we prove an estimate for $\sup_{1\le t\le T}||D^ku(t)||_2$.
\begin{lemma}[Estimate for $||D^ku(t)||_2$]\label{le:D^k estimate for d>=3 small}
We have
\begin{equation}
||D^ku(t)||_2\le C^{DS}||u_1||_{H^k}+C^{kE}||w||_1M(T)^3\, , 
\end{equation}
where $C^{kE}:=4C^{ES}C^{DS}(C^{KP})^2\frac{1}{d-1}$. 
\end{lemma}
\begin{proof}
By Duhamel's formula, we have
\begin{equation}
u(t)= e^{-\ri (t-1)H}u_1-\ri\int_1^t\dr s e^{-\ri (t-s)H}(w*|u(s)|^2)u(s)\, ,
\end{equation}
so by applying $D^k$ to both sides, we get
\begin{equation}
D^ku(t)= D^ke^{-\ri (t-1)H}u_1-\ri\int_1^t\dr s D^ke^{-\ri (t-s)H}(w*|u(s)|^2)u(s)\, . 
\end{equation}
By Lemma \ref{le:W^(k,p) dispersive estimate} for $p=2$, we obtain
\begin{align*}
&\quad ||D^ku(t)||_2\le ||D^ke^{-\ri (t-1)H}u_1||_2+\int_1^t\dr s ||D^ke^{-\ri (t-s)H}(w*|u(s)|^2)u(s)||_2\\
&\le ||e^{-\ri (t-1)H}u_1||_{H^k}+\int_1^t\dr s ||e^{-\ri (t-s)H}(w*|u(s)|^2)u(s)||_{H^k}\\
&\le C^{DS}\left(||u_1||_{H^k}+\int_1^t\dr s ||(w*|u(s)|^2)u(s)||_{H^k}\right)\\
&\le C^{DS}\Bigg(||u_1||_{H^k}+\int_1^t\dr s C^{ES}\left(||(w*|u(s)|^2)u(s)||_2+||D^k[(w*|u(s)|^2)u(s)]||_2\right)\Bigg)\\
&\le C^{DS}||u_1||_{H^k}+C^{ES}C^{DS}\int_1^t\dr s \left(||(w*|u(s)|^2)u(s)||_2+||D^k[(w*|u(s)|^2)u(s)]||_2\right)
\end{align*}
for a constant $C^{ES}\ge 1$. From the proof of Lemma \ref{le:Sobolev type estimate in d>=3 small}, we know that
\begin{equation}
||(w*|u(s)|^2)u(s)||_2\le ||w||_1 M(T)^3|s|^{-d}
\end{equation}
and 
\begin{equation}
||D^k[(w*|u(s)|^2)u(s)]||_2\le 3(C^{KP})^2||w||_1M(T)^3|s|^{-d}\, . 
\end{equation}
\subsubsection*{Conclusion.}
To sum up, we have
\begin{align*}
&\quad ||D^ku(t)||_2\\
&\le C^{DS}||u_1||_{H^k}+C^{ES}C^{DS}\int_1^t\dr s \left(||(w*|u(s)|^2)u(s)||_2+||D^k[(w*|u(s)|^2)u(s)]||_2\right)\\
&\le C^{DS}||u_1||_{H^k}+C^{ES}C^{DS}\int_1^t\dr s \left(||w||_1 M(T)^3|s|^{-d}+3(C^{KP})^2||w||_1M(T)^3|s|^{-d}\right)\\
&\le C^{DS}||u_1||_{H^k}+4C^{ES}C^{DS}(C^{KP})^2||w||_1M(T)^3\int_1^t\dr s|s|^{-d}\\
&\le C^{DS}||u_1||_{H^k}+4C^{ES}C^{DS}(C^{KP})^2\frac{1}{d-1}||w||_1M(T)^3\\
&\le C^{DS}||u_1||_{H^k}+C^{kE}||w||_1M(T)^3\, , 
\end{align*}
where we define $C^{kE}:=4C^{ES}C^{DS}(C^{KP})^2\frac{1}{d-1}$.
\end{proof}

Figure \ref{fi:graph of f} illustrates the result of the following small technical lemma.
\begin{lemma}\label{le:f>=0 two intervals}
Let $C>0$. Then there exists $\epsilon>0$ such that the function
\begin{equation}
f:[0,\infty)\to\R, f(x):=\epsilon+Cx^3-x
\end{equation}
satisfies the following: $\{f\ge 0\}$ consists of two disjoint intervals $I_1,I_2$ that have a strictly positive distance from each other, where we choose $I_1$ such that $0\in I_1$. Moreover, $I_1$ is bounded. 
\end{lemma}
\begin{proof}
For every $\epsilon>0$, we have
\begin{equation}
f'(x)=3Cx^2-1\, . 
\end{equation}
Thus, $f'$ is a smooth function that is strictly increasing and it satisfies $f'(0)=-1$ and $\lim_{x\to\infty}f'(x)=\infty$. It follows that $f$ has at most two zeroes. Since $f(0)=\epsilon>0$ and $\lim_{x\to\infty}f(x)=\infty$, we are done if we can show that there exists a point $\tilde x\in[0,\infty)$ such that $f(\tilde x)<0$. 
Note that for $x\ge 0$, we have
\begin{equation}
-\frac{1}{2}\ge f'(x)=3Cx^2-1\iff\frac{1}{2}\ge 3Cx^2\iff\frac{1}{6C}\ge x^2\iff\frac{1}{\sqrt{6C}}\ge x\, . 
\end{equation}
We define $\tilde x:=\frac{1}{\sqrt{6C}}$. Now, choose $\epsilon>0$ such that $\epsilon<\frac{1}{2\sqrt{6C}}$. We obtain
\begin{equation}
f(\tilde x)=f(0)+\int_0^{\tilde x}\dr t f'(t)\le\epsilon+\int_0^{\tilde x}\frac{-1}{2}=\epsilon-\frac{1}{2}\cdot \frac{1}{\sqrt{6C}}=\epsilon-\frac{1}{2\sqrt{6C}}<0\, . 
\end{equation} 
\end{proof}

\section{Conclusion of the main theorem}\label{s:conclusion of the main theorem}
In this section, we prove Theorem \ref{th:dispersive d>=3 small} using the estimates from Section \ref{s:various estimates}. Furthermore, we show an extension of Theorem \ref{th:dispersive d>=3 small} for large initial data under certain additional assumptions and we explain two proof strategies for Theorem \ref{th:cubic NLS main result}. 
\begin{proof}[Proof of Theorem \ref{th:dispersive d>=3 small}]
Let us work with the Hartree type equation with initial data at $t=1$. We decompose the proof into two parts. In the first part, we prove the decay estimate
\begin{equation}
||u(t)||_\infty\le C_0|t|^{-\frac{d}{2}}
\end{equation}
for all $t\ge 1$. In the second part, we show 
\begin{equation}
||\partial_t u(t)||_\infty\le \tilde C_0 |t|^{-\frac{d}{2}}
\end{equation}
for all $t\ge 1$. 
\\ \\
{\bf Part 1: $||u(t)||_\infty\le C_0|t|^{-\frac{d}{2}}$.}
Define 
\begin{equation}
C:=3||w||_1\max\left\{ C^{\infty E} , C^{kE}  \right\}\, . 
\end{equation}
By Lemma \ref{le:f>=0 two intervals}, there exists $\epsilon>0$ small enough such that the function 
\begin{equation}
f:[0,\infty)\to\R,\ f(x):=\epsilon+Cx^3-x
\end{equation}
satisfies the following: $\{f\ge 0\}$ consists of two intervals $I_1,I_2$ that have a strictly positive distance from each other and $0\in I_1$. Moreover, $I_1$ is bounded. We fix such an $\epsilon>0$. Let $C_0:=\sup I_1>0$ be the first zero of $f$. We define
\begin{equation}
\epsilon_0:=\frac{\min\{\epsilon, C_0\}}{3C^VC^{DS}}\, . 
\end{equation}
Thus, if the assumption 
\begin{equation}
 ||e^{\ri H}u_1||_1\, , ||u_1||_{H^k}\le\epsilon_0
\end{equation}
is satisfied, we know by $C^V, C^{DS}\ge 1$ that
\begin{equation}
C^V||e^{\ri H}u_1||_1\le\frac{\epsilon}{3}\, ,\
C^{DS}||u_1||_{H^k}\le\frac{\epsilon}{3}\, ,\
||u_1||_2\le ||u_1||_{H^k}\le\frac{\epsilon}{3}\, . 
\end{equation}
Moreover, we also have
\begin{align*}
M(1)&= ||u_1||_\infty+||D^ku_1||_2 +||u_1||_2=||e^{-\ri H}e^{\ri H}u_1||_\infty+||D^ku_1||_2 +||u_1||_2\\
&\le C^V1^{-\frac{d}{2}}||e^{\ri H}u_1||_1+||u_1||_{H^k} +||u_1||_2\le  C^V\epsilon_0+\epsilon_0+\epsilon_0\le 3C^V\epsilon_0\le C_0\, .
\end{align*}
Let $T\ge 1$. By Definition \ref{de:M(T) for d>=3 small}, Corollary \ref{co:estimate infinity norm d>=3 small} and Lemma \ref{le:D^k estimate for d>=3 small}, we have
\begin{align*}
M(T)&=\sup_{1\le t\le T} |t|^{\frac{d}{2}}||u(t)||_\infty+\sup_{1\le t\le T}||D^ku(t)||_2 +||u_1||_2\\
&\le C^V||e^{\ri H}u_1||_1+C^{\infty E}||w||_1M(T)^3+C^{DS}||u_1||_{H^k}+C^{kE}||w||_1M(T)^3+||u_1||_2\\
&\le \frac{\epsilon}{3}+\frac{C}{3}M(T)^3+\frac{\epsilon}{3}+\frac{C}{3}M(T)^3+\frac{\epsilon}{3}\le \epsilon+CM(T)^3\, . 
\end{align*}
Therefore, 
\begin{equation}\label{eq:epsilon+CM(T)^3-M(T)>=0 in proof}
\epsilon+CM(T)^3-M(T)\ge 0
\end{equation}
for all $T\ge 1$ with $M(T)<\infty$. 
\\ \\
We can now conclude as we explained at the end of Subsection \ref{ss:proof strategy introduction}. By Theorem \ref{th:H^m solutions for m>d/2}, we know that $u\in C\left([1,T_{\max}), H^k(\R^d)\right)$ for some $T_{\max}\in (1,\infty]$. By the Sobolev embedding theorem, $T\mapsto M(T)$ is continuous on $[1,T_{\max})$. Thus, by the choice of $\epsilon>0$ and \eqref{eq:epsilon+CM(T)^3-M(T)>=0 in proof}, we deduce that $M(T)\le C_0$ for all $T\in[1,T_{\max})$. By the blow-up alternative in Theorem \ref{th:H^m solutions for m>d/2}, we obtain $T_{\max}=\infty$. Therefore, $M(T)\le C_0$ for all $T\ge 1$. In particular,
\begin{equation}
||u(t)||_\infty\le\frac{C_0}{|t|^{\frac{d}{2}}}
\end{equation}
for all $t\ge 1$, which is the desired result. 
\\ \\
{\bf Part 2: $||\partial_t u(t)||_\infty\le \tilde C_0 |t|^{-\frac{d}{2}}$.}
Let $\epsilon, C, C_0>0$ be as in the proof of part 1. Our proof strategy is to define a quantity $\tilde M(T)$ similar to $M(T)$, which contains $\sup_{1\le t\le T}|t|^{\frac{d}{2}}||\partial_t u(t)||_\infty$. Our goal is to prove a bound of the form 
\begin{equation}
\tilde M(T)\le \tilde \epsilon+\tilde C (\tilde M(T))^3
\end{equation}
and to argue as before. It will be essential to make sure that $M(T)$ is small for all $T\ge 1$.
\subsubsection*{Boundedness of $||\partial_t u||_2$}
The Hartree type equation is
\begin{equation}\label{eq:Hartree type equation for u in proof for derivative of u}
\ri\partial_t u  =-\Delta u+ Vu+(w*|u|^2)u\, .
\end{equation}
Thus, for every $t\ge 1$, we have
\begin{align*}
||\partial_t u(t)||_2&\le ||\Delta u(t)||_2+ ||V||_\infty||u_1||_2+||w*|u(t)|^2||_\infty||u_1||_2\\
&\le (||D^ku(t)||_2+||u_1||_2)+ ||V||_\infty||u_1||_2+||w||_1||u(t)||_\infty^2||u_1||_2\\
&\le (C_0+C_0)+||V||_\infty C_0+||w||_1C_0^3\le C_0(2+||V||_\infty+||w||_1C_0^2)\, ,
\end{align*}
where we used $k\ge 2$, $||u(t)||_\infty^2\le \left( C_0|t|^{-\frac{d}{2}}\right)^2\le C_0^2$ since $t\ge 1$ and
\begin{equation}
||\Delta u||_2^2\le ||u||_2^2+||D^ku||_2^2\le\left(||u||_2+||D^ku||_2\right)^2\, .
\end{equation}
In particular, $||\partial_t u(t)||_2$ is bounded by a constant, which is small if $C_0$ is small.
\subsubsection*{Duhamel's formula for $\partial_t u(t)$.}
Differentiating the Hartree type equation with respect to time, we get
\begin{equation}\label{eq:Hartree type equation for the time derivative of u}
\ri\partial_t (\partial_t u)  =(-\Delta + V)(\partial_t u)+\partial_t[(w*|u|^2)u]\, .
\end{equation}
We can differentiate the Hartree type equation because
\begin{equation}
u\in C\left([1,\infty),H^k(\R^d)\right)\cap C^1\left([1,\infty),H^{-1}(\R^d)\right)\, .
\end{equation}
Note that \eqref{eq:Hartree type equation for the time derivative of u} holds for every $t\in [1,\infty)$. 
\\ \\
We would like to apply Duhamel's formula for $\partial_t u$. To this end, we need to show that $\partial_t u\in C\left([1,\infty), L^2(\R^d)\right)$ and $\partial_t[(w*|u|^2)u]\in C\left([1,\infty), H^{-2}(\R^d)\right)$, see Remark \ref{re:generalised Duhamel} and \cite[Remark 1.6.1(ii)]{cazenave}.
\\ \\
Let us start by showing that $\partial_t u\in C\left([1,\infty), L^2(\R^d)\right)$. The fact that $-\Delta u(t)+Vu(t)$ belongs to $C\left([1,\infty), L^2(\R^d)\right)$ follows from $u\in C\left([1,\infty), H^k(\R^d)\right)$. Lemma \ref{le:H^m boundedness and continuity of interaction part}(ii) together with $u\in C\left([1,\infty), H^k(\R^d)\right)$ and $\sup_{t\ge 1}||u(t)||_\infty <\infty$ imply that $(w*|u|^2)u\in C\left([1,\infty), L^2(\R^d)\right)$. Therefore, $\partial_t u\in C\left([1,\infty), L^2(\R^d)\right)$ holds true because $u$ satisfies the Hartree type equation \eqref{eq:Hartree type equation for u in proof for derivative of u}. 
\\ \\
We have 
\begin{equation}
\partial_t[(w*|u|^2)u]=(w*(\overline{u}\partial_t u+u\overline{\partial_t u}))u+(w*|u|^2)\partial_tu\, .
\end{equation}
Now, similar to the proof of Lemma \ref{le:H^m boundedness and continuity of interaction part}(ii), we can show that
\begin{equation}
\partial_t[(w*|u|^2)u]\in C\left([1,\infty), L^2(\R^d)\right)
\end{equation}
using $\partial_t u\in C\left([1,\infty), L^2(\R^d)\right)$ and $u\in C\left([1,\infty), H^k(\R^d)\right)$. We omit a detailed computation here. 
\\ \\
Therefore, we can apply Duhamel's formula to \eqref{eq:Hartree type equation for the time derivative of u} to get
\begin{equation}
(\partial_t u)(t)=e^{-\ri (t-1)H}(\partial_t u)(1)-\ri\int_1^t\dr s e^{-\ri (t-s)H}\partial_t[(w*|u|^2)u](s)\, . 
\end{equation} 
\subsubsection*{Definition of $\tilde M(T)$.}
For every $T\ge 1$, define
\begin{equation}
\tilde M(T):= M(T)+ \sup_{1\le t\le T} |t|^{\frac{d}{2}}||(\partial_t u)(t)||_\infty+\sup_{1\le t\le T}||D^k(\partial_t u)(t)||_2+\sup_{1\le t\le T}||(\partial_t u)(t)||_2\, .
\end{equation}
\subsubsection*{Estimates for $\sup_{1\le t\le T} |t|^{\frac{d}{2}}||(\partial_t u)(t)||_\infty$ and $\sup_{1\le t\le T}||D^k(\partial_t u)(t)||_2$.}
Using the same type of estimates as in the proof of Theorem \ref{th:dispersive d>=3 small},  we can estimate
\begin{equation}
\sup_{1\le t\le T} |t|^{\frac{d}{2}}||(\partial_t u)(t)||_\infty\le C^V||e^{\ri H}(\partial_t u)(1)||_1+\frac{\tilde C}{3}\tilde M(T)^3
\end{equation}
and
\begin{equation}
\sup_{1\le t\le T}||D^k(\partial_t u)(t)||_2\le  C^{DS}||(\partial_t u)(1)||_{H^k}+\frac{\tilde C}{3}\tilde M(T)^3
\end{equation}
for some $\tilde C>0$, which only depends on $d,V,||w||_1$. Thus, if
\begin{align*}
\sup_{1\le t\le T} M(T) + \sup_{1\le t\le T}||(\partial_t u)(t)||_2\le C_0(3+||V||_\infty+||w||_1C_0^2)&\le\frac{\tilde \epsilon}{3}\, ,\\
C^V||e^{\ri H}(\partial_t u)(1)||_1&\le\frac{\tilde \epsilon}{3}\, ,\\
C^{DS}||(\partial_t u)(1)||_{H^k}&\le \frac{\tilde \epsilon}{3}\, ,
\end{align*}
where we used $M(T)\le C_0$ and the bound for $||\partial_tu||_2$, we get
\begin{equation}
\tilde M(T)\le \tilde \epsilon+\tilde C (\tilde M(T))^3
\end{equation}
for all $T\ge 1$. 
\subsubsection*{Conclusion.}
Choose $\tilde\epsilon>0$ small enough such that $\{\tilde f\ge 0\}$ consists of two disjoint closed intervals, where $\tilde f:[0,\infty)\to\R,\ x\mapsto \tilde\epsilon+\tilde Cx^3-x$. Fix this $\tilde\epsilon>0$. Now choose $\epsilon>0$ small enough, and hence $C_0=C_0(\epsilon,C)>0$ small enough such that $\{f\ge 0\}$ consists of two disjoint closed intervals and 
\begin{equation}
C_0(3+||V||_\infty+||w||_1C_0^2)\le\frac{\tilde \epsilon}{3}\, . 
\end{equation}
Here, we used that $C_0>0$ is small if $\epsilon>0$ is small. So far, we have fixed $\epsilon,C_0,\tilde \epsilon,\tilde C_0$. Now suppose that the initial data satisfies
\begin{equation}
C^V||e^{\ri H}u_1||_1\le\frac{\epsilon}{3}\, , \
C^{DS}||u_1||_{H^k}\le \frac{\epsilon}{3}\, ,\
M(1)\le C_0
\end{equation}
and 
\begin{equation}
C^V||e^{\ri H}(\partial_t u)(1)||_1\le\frac{\tilde \epsilon}{3}\, ,\
C^{DS}||(\partial_t u)(1)||_{H^k}\le \frac{\tilde \epsilon}{3}\, ,\
\tilde M(1)\le \tilde C_0\, .
\end{equation}
Note that by the Sobolev embedding theorem, we have
\begin{equation}
||u_1||_\infty\le C^S||u_1||_{H^k}
\end{equation}
and similarly for $||(\partial_t u)(1)||_\infty$. Thus, $M(1)$ and $\tilde M(1)$ can be controlled by $||u_1||_{H^k}$ and $||(\partial_t u)(1)||_{H^k}$. In a similar way to the proof of \cite[Theorem 4.10.1]{cazenave} and Theorem \ref{th:H^m solutions for m>d/2}, we can show that there exists $T_{\max}\in(1,\infty]$ such that $\partial_tu\in C\left([1,T_{\max}),H^k(\R^d)\right)$ and the corresponding blow-up alternative holds. We can now argue as in part 1 to obtain
\begin{equation}
\tilde M(T)\le\tilde C_0
\end{equation}
for all $T\ge 1$, which is the desired result. 
\end{proof}
\begin{remark}[Extension of Theorem \ref{th:dispersive d>=3 small} for large initial data under additional assumptions]\label{re:large d>=3}
The proof of Theorem \ref{th:dispersive d>=3 small} shows that if, in addition, we know that
\begin{equation}
\sup_{t\ge 0}||D^ku(t)||_\infty <\infty
\end{equation}
and 
\begin{equation}
\lim_{t\to\infty}||u(t)||_\infty=0\, ,
\end{equation}
then we can also show the decay estimate
\begin{equation}
||u(t)||_\infty\le\frac{C_0}{(1+|t|)^{\frac{d}{2}}}\  \mathrm{for\  all}\ t\ge 0
\end{equation}
for large initial data. That is, we do not need any smallness condition on $||e^{-\ri H}u_1||_1$, $||u_1||_{H^k}$. The proof which we present in this remark follows closely the proof strategy of \cite[Corollary 3.4]{gm}. In certain circumstances, it might be known a priori that our additional assumptions are satisfied. For instance, this is the case in \cite{gm}, see the beginning of Section 3 and Proposition 3.3 there.
\\ \\
Again, let us work in the setting of the Hartree type equation with initial data at $t=1$. Define
\begin{equation}
N(t):=\sup_{1\le r\le t}|r|^{\frac{d}{2}}||u(r)||_\infty\, .
\end{equation}
By the proof of the direct estimate (Lemma \ref{le:direct estimate d>=3 small}), the Sobolev type estimate (Lemma \ref{le:Sobolev type estimate in d>=3 small}) and our additional assumptions, we know that there exists a constant 
\begin{equation}
C_1=C_1\left( d,\,  C^S,\,  C^{ES},\, C^V,\, C^{DS},\, C^{KP},\, ||w||_1,\, ||u_1||_2,\,\sup_{t\ge 1}||D^ku(t)||_2\right)>0
\end{equation}  
such that
\begin{equation}\label{eq:direct estimate for remark for large data in d>=3}
||e^{-\ri (t-s)H}(w*|u(s)|^2)u(s)||_\infty\le C_1|t-s|^{-\frac{d}{2}}N(s)|s|^{-\frac{d}{2}}
\end{equation}
and
\begin{equation}\label{eq:Sobolev type estimate for remark for large data in d=>3}
||e^{-\ri (t-s)H}(w*|u(s)|^2)u(s)||_\infty\le C_1 ||u(s)||_\infty N(s)|s|^{-\frac{d}{2}}\, . 
\end{equation}
Combining (\ref{eq:direct estimate for remark for large data in d>=3}) and (\ref{eq:Sobolev type estimate for remark for large data in d=>3}), we get
\begin{equation}\label{eq:Combined direct and Sobolev type estimate for remark for large data in d=>3}
||e^{-\ri (t-s)H}(w*|u(s)|^2)u(s)||_\infty\le C_1||u(s)||_\infty^{\frac{1}{4}}|t-s|^{-\frac{3d}{8}}N(s)|s|^{-\frac{d}{2}}\, . 
\end{equation}
Note that $\frac{3d}{8}>1$ since $d\ge 3$ and define $C_d:=\int_1^\infty \dr s |s|^{-\frac{3d}{8}}<\infty$. If $1\le t\le 2$, let $t_1:=t_2:=1$. If $t>2$, let $t_1:=\frac{t}{2}$ and $t_2:=t-1$. Thus, we always have $1\le t_1\le t_2\le t$.  We use (\ref{eq:direct estimate for remark for large data in d>=3}) for $1\le s\le t_1$, (\ref{eq:Combined direct and Sobolev type estimate for remark for large data in d=>3}) for $t_1\le s\le t_2$ and (\ref{eq:Sobolev type estimate for remark for large data in d=>3}) for $t_2\le s\le t$ to get
\begin{align*}
& |t|^{\frac{d}{2}}||u(t)||_\infty \le |t|^{\frac{d}{2}} ||e^{-\ri (t-1)H}u_1||_\infty+|t|^{\frac{d}{2}}\int_1^{t_1}\dr s ||e^{-\ri (t-s)H}(w*|u(s)|^2)u(s)||_\infty\\
&\qquad+|t|^{\frac{d}{2}}\int_{t_1}^{t_2}\dr s ||e^{-\ri (t-s)H}(w*|u(s)|^2)u(s)||_\infty\\
&\qquad +|t|^{\frac{d}{2}}\int_{t_2}^t\dr s ||e^{-\ri (t-s)H}(w*|u(s)|^2)u(s)||_\infty\\
&\le C^V||e^{\ri H}u_1||_1+C_1 \int_1^{t_1}\dr s |t|^{\frac{d}{2}}|t-s|^{-\frac{d}{2}}N(s)|s|^{-\frac{d}{2}}\\
&\qquad +C_1\sup_{t_1\le r\le t_2}||u(r)||_\infty^{\frac{1}{4}}\int_{t_1}^{t_2}\dr s |t-s|^{-\frac{3d}{8}} N(s)|t|^{\frac{d}{2}}|s|^{-\frac{d}{2}}\\
&\qquad +C_1 \sup_{t_2\le r\le t} ||u(r)||_\infty\int_{t_2}^t\dr s  N(s)|t|^{\frac{d}{2}}|s|^{-\frac{d}{2}}\\
&\le C^V||e^{\ri H}u_1||_1+2^{\frac{d}{2}}C_1 \int_1^{t_1}\dr s N(s)|s|^{-\frac{d}{2}}+2^{\frac{d}{2}}C_1N(t)\sup_{t_1\le r\le t}||u(r)||_\infty^{\frac{1}{4}}\int_{t_1}^{t_2}\dr s |t-s|^{-\frac{3d}{8}} \\
&\qquad  +2^{\frac{d}{2}}C_1 N(t)\sup_{t_1\le r\le t} ||u(r)||_\infty\\
&\le C^V||e^{\ri H}u_1||_1+2^{\frac{d}{2}}C_1 \int_1^{t_1}\dr s N(s)|s|^{-\frac{d}{2}} \\
&\qquad+2^{\frac{d}{2}}C_1(C_d+1)\sup_{t_1\le r\le t}\left(||u(r)||_\infty^{\frac{1}{4}}+||u(r)||_\infty\right) N(t)\, .
\end{align*}
Fix $T_0\ge 2$ large enough such that
\begin{equation}
2^{\frac{d}{2}}C_1(C_d+1)\sup_{\frac{T_0}{2}\le r}\left(||u(r)||_\infty^{\frac{1}{4}}+||u(r)||_\infty\right)\le\frac{1}{2}\, .
\end{equation}
Note that this is possible because $\lim_{t\to\infty}||u(t)||_\infty=0$. Let $T\ge T_0$. By taking the supremum over $1\le t\le T$, we get
\begin{align*}
&\qquad N(T)=\sup_{1\le r\le T}|r|^{\frac{d}{2}}||u(r)||_\infty\le \sup_{1\le r\le T_0}|r|^{\frac{d}{2}}||u(r)||_\infty+\sup_{T_0\le r\le T}|r|^{\frac{d}{2}}||u(r)||_\infty\\
&\le N(T_0)+C^V||e^{\ri H}u_1||_1+2^{\frac{d}{2}}C_1 \int_1^{T}\dr s N(s)|s|^{-\frac{d}{2}}+\frac{1}{2}N(T)\\
&\le N(T_0)+C^V||e^{\ri H}u_1||_1+2^{\frac{d}{2}}C_1 \int_1^{T_0}\dr s N(s)|s|^{-\frac{d}{2}}+2^{\frac{d}{2}}C_1 \int_{T_0}^T\dr s N(s)|s|^{-\frac{d}{2}}+\frac{1}{2}N(T)\\
&\le N(T_0)+C^V||e^{\ri H}u_1||_1+2^{\frac{d}{2}}C_1 \frac{1}{\frac{d}{2}-1}N(T_0)+2^{\frac{d}{2}}C_1 \int_{T_0}^T\dr s N(s)|s|^{-\frac{d}{2}}+\frac{1}{2}N(T)\\
&\le \left(1+2^{\frac{d}{2}}C_1 \frac{2}{d-2}\right) N(T_0)+C^V||e^{\ri H}u_1||_1+2^{\frac{d}{2}}C_1 \int_{T_0}^T\dr s N(s)|s|^{-\frac{d}{2}}+\frac{1}{2}N(T)\, . 
\end{align*}
Therefore, since  $N(T)<\infty$ for every $T\ge 1$, we have
\begin{equation}
N(T)\le 2\left(1+2^{\frac{d}{2}}C_1 \frac{2}{d-2}\right) N(T_0)+2C^V||e^{\ri H}u_1||_1+2^{\frac{d}{2}+1}C_1 \int_{T_0}^T\dr s N(s)|s|^{-\frac{d}{2}}\, .
\end{equation}
We can now apply Gronwall's inequality, see \cite[Lemma 2.7]{teschl}, with
\begin{align*}
I&:= [T_0,\infty)\\
\alpha &:= 2\left(1+2^{\frac{d}{2}}C_1 \frac{2}{d-2}\right) N(T_0)+2C^V||e^{\ri H}u_1||_1\\
\beta(s)&:=2^{\frac{d}{2}+1}C_1|s|^{-\frac{d}{2}}
\end{align*}
to get
\begin{equation}
N(T)\le \alpha e^{\int_{T_0}^T\dr s\beta(s)}\le \alpha e^{||\beta||_{L^1(I)}}=:C_0<\infty\, ,
\end{equation}
where we used $d\ge 3$ to deduce that $||\beta||_{L^1(I)}<\infty$. This shows that
\begin{equation}
\sup_{t\ge 1}|t|^{\frac{d}{2}}||u(t)||_\infty\le C_0\, .
\end{equation}
\end{remark}

\begin{proof}[Proof of Theorem \ref{th:cubic NLS main result}]
We explain two different proof strategies.
\\ \\
{\bf Adaptation of the proof strategy for Theorem \ref{th:dispersive d>=3 small}.} We can prove Theorem \ref{th:cubic NLS main result} in a very similar way to Theorem \ref{th:dispersive d>=3 small}. We need results for the cubic nonlinear Schrödinger equation that are similar to the results we mentioned in Subsection \ref{ss:hartree type equation}. The preliminaries we need for the cubic nonlinear Schrödinger equation are Duhamel's formula for both $u$ and $\partial_t u$, a theorem on $H^k$-solutions similar to Theorem \ref{th:H^m solutions for m>d/2} and the conservation of mass on the maximal time interval of existence of the $H^k$-solution $u$. 
\\ \\
For the theorem on $H^k$-solutions, we argue as in the proof of Theorem \ref{th:H^m solutions for m>d/2}: Recall that we can closely follow the proof of \cite[Theorem 4.10.1(i)]{cazenave} but we have to make sure that the estimates in Lemma \ref{le:H^m boundedness and continuity of interaction part} are satisfied for our nonlinearity $Vu\pm |u|^2 u$ and that there is uniqueness of $H^k$-solutions. The assumption that $g(z)=\pm |z|^2 z$ belongs to $C^k(\C,\C)$ in the real sense and $g(0)=0$ is satisfied, so we can use \cite[Lemma 4.10.2]{cazenave} for this part. For the term $Vu$, we use Lemma \ref{le:H^m boundedness and continuity of interaction part}. The uniqueness follows from \cite[Proposition 4.2.9]{cazenave} with $s=k$, $r_j=\rho_j=2$ and $q_j=\infty$ using that $(\infty,2)$ is an admissible pair by \cite[Definition 3.2.1]{cazenave} and the Sobolev inequality $||f||_\infty\le C^S ||f||_{H^k}$. The conservation of the $L^2$ norm can be obtained in the same way as in \cite[Theorem 4.10.1(iii)]{cazenave}.
\\ \\
Duhamel's formula holds for the $H^k$-solution to the cubic nonlinear Schrödinger equation by \cite[Remark 1.6.1(ii)]{cazenave}: If $I$ is the maximal time interval of existence, then $u\in C\left(I,H^k(\R^d)\right)\subset C\left(I,L^2(\R^d)\right)$. Furthermore, $\pm |u|^2u\in C\left(I, L^2(\R^d)\right)\subset C\left(I,H^{-2}(\R^d)\right)$ and $u_0\in H^k(\R^d)\subset L^2(\R^d)$. Moreover, we can also apply Duhamel's formula to $\partial_tu$: We can differentiate the right-hand side of the equation, and therefore also the left-hand side, to get
\begin{equation}
\ri\partial_t(\partial_t u)=(-\Delta+V)\partial_t \pm\partial_t\left(|u|^2u\right)\, .
\end{equation}
By the cubic nonlinear Schrödinger equation, we know that $\partial_t u\in C\left(I, L^2(\R^d)\right)$ and thus, $\pm\partial_t\left(|u|^2u\right)\in C\left(I, L^2(\R^d)\right)\subset C\left(I,H^{-2}(\R^d)\right)$. Therefore, we also have Duhamel's formula for $\partial_t u$. 
\\ \\
Using these facts, we can deduce estimates that are very similar to the estimates in Section \ref{s:various estimates} since we only use the fact that $||w||_1<\infty$ there. For the cubic nonlinear Schrödinger equation, we can think of $w$ as $w=\pm \delta_0$, which formally also satisfies this property. Using the blow-up criterion from the theorem on $H^k$-solutions to the cubic nonlinear Schrödinger equation, we can conclude as we explained in Subsection \ref{ss:proof strategy introduction}.  
\\ \\
{\bf Implication from Theorem \ref{th:dispersive d>=3 small} by considering the cubic nonlinear Schrödinger equation as a limit of Hartree type equations.}
For simplicity, let us focus on the defocusing cubic nonlinear Schrödinger equation. The proof for the focusing case works in the same way by replacing the plus sign in front of the interaction term by a minus sign. Using the theorem on $H^k$-solutions for the cubic nonlinear Schrödinger equation, which we explained above, we know that there exists a maximal time interval of existence $I\subset\R$ of the cubic nonlinear Schrödinger equation
\begin{equation} 
\begin{cases} 
\ri\partial_t u & =-\Delta u+ Vu+|u|^2u\\
u(0) & =u_0\, .
\end{cases}
\end{equation}
Moreover, the solution $u$ satisfies $u\in C\left(I,H^k(\R^d)\right)$, the $L^2$ norm of $u$ is conserved in time and the blow-up alternative holds. Fix $w\in C_c^\infty(\R^d)$ with $w\ge 0$, $\int_{\R^d} w=1$ and define for $n\in\N$ the function
\begin{equation}
w_n(x):=n^dw(nx)\, .
\end{equation}
Note that $||w_n||_1=1$ for all $n\in\N$ and that $w_n$ converges to the delta distribution in the distributional sense. For every $n\in\N$ let $u_n$ be the solution to the Hartree type equation
\begin{equation} 
\begin{cases} 
\ri\partial_t u_n & =-\Delta u_n+ Vu_n+(w_n*|u_n|^2)u_n\\
u_n(0) & =u_0\, .
\end{cases}
\end{equation}
Fix $T>0$ with $T\in I$. For every $n\in\N$ and every $t\in[0,T]$, we have
\begin{equation}\label{eq:derivative of L2norm of u-u_n first equality}
\frac{\rd}{\rd t}||(u-u_n)(t)||_2^2=2\im \langle u-u_n, |u|^2u-(w_n*|u_n|^2)u_n\rangle\, , 
\end{equation}
where we used that $-\Delta+V$ is self-adjoint. We omit the $t$-dependence in these computations for simplicity of notation. Let us split this term into two parts
\begin{align*}
&\qquad \im \langle u-u_n, |u|^2u-(w_n*|u_n|^2)u_n\rangle\\
&=  \im \langle u-u_n, \left(|u|^2-w_n*|u|^2\right)u\rangle+\im \langle u-u_n, \left(w_n*(|u|^2-|u_n|^2)\right)u\rangle\\
&\qquad + \im \langle u-u_n, \left(w_n*|u_n|^2\right)(u-u_n)\rangle\\
&=  \im \langle u-u_n, \left(|u|^2-w_n*|u|^2\right)u\rangle+\im \langle u-u_n, \left(w_n*(|u|^2-|u_n|^2)\right)u\rangle =:(I)+(II)\, , 
\end{align*}
where we used that the third term vanishes because $w_n*|u_n|^2$ is real-valued. We estimate both terms separately:
\begin{align*}
|(I)|=&\left|\im \langle u-u_n, \left(|u|^2-w_n*|u|^2\right)u\rangle\right|=\left|\im \langle -u_n, \left(|u|^2-w_n*|u|^2\right)u\rangle\right|\\
&\le ||u_n||_2|||u|^2-w_n*|u|^2||_\infty||u||_2 = ||u_0||_2^2|||u|^2-w_n*|u|^2||_\infty\, .
\end{align*}
For fixed $t\in[0,T]$, we know that $u\in H^k(\R^d)\subset C_0(\R^d)$. In particular, $|u|^2$ is uniformly continuous, and therefore, we know that 
\begin{equation}\label{eq:|||u|^2-w_n*|u|^2|| converges to 0}
\lim_{n\to\infty}|||u|^2-w_n*|u|^2||_\infty=0\, .
\end{equation}
For the second term, we have
\begin{align*}
|(II)|&= \left| \im \langle u-u_n, \left(w_n*(|u|^2-|u_n|^2)\right)u\rangle\right|\le ||u-u_n||_2||w_n*(|u|^2-|u_n|^2)||_2||u||_\infty\\
&\le ||u-u_n||_2||u\overline{(u-u_n)}+\overline{u_n}(u-u_n)||_2||u||_\infty\le ||u-u_n||_2^2(||u||_\infty+||u_n||_\infty)||u||_\infty\, .
\end{align*}
Recall that by the dispersive estimate from Theorem \ref{th:dispersive d>=3 small}, we have
\begin{equation}
||u_n(t)||_\infty\le \frac{C_0}{(1+|t|)^{\frac{d}{2}}}\, ,
\end{equation}
where the constant $C_0$ is uniform in $n$ because $||w_n||_1=1$ for all $n\in\N$. Moreover, recall that $||u(t)||_{H^k}$ is bounded uniformly in $t\in[0,T]$ because $u\in C\left(I,H^k(\R^d)\right)$. Thus, we know that there exists a constant $C=C(T)>0$ such that
\begin{equation}
|(II)|\le \frac{C}{2} ||u-u_n||_2^2\, . 
\end{equation}
By \eqref{eq:derivative of L2norm of u-u_n first equality}, we obtain
\begin{equation}
\left|\frac{\rd}{\rd t}||(u-u_n)||_2^2\right|\le 2||u_0||_2^2|||u|^2-w_n*|u|^2||_\infty+C||u-u_n||_2^2\, .
\end{equation}
Therefore, using $u(0)=u_0=u_n(0)$, we obtain
\begin{align*}
||(u-u_n)(t)||_2^2&\le ||(u-u_n)(0)||_2^2+\int_0^t\dr s\left|\left(\frac{\rd}{\rd t}||(u-u_n)||_2^2\right) (s)\right|\\
&\le 2||u_0||_2^2\int_0^T\dr s|||u(s)|^2-w_n*|u(s)|^2||_\infty+C\int_0^t\dr s||(u-u_n)(s)||_2^2\, .
\end{align*}
Define for every $n\in\N$
\begin{equation}
\epsilon_n:=\epsilon_n(T):= 2||u_0||_2^2\int_0^T\dr s|||u(s)|^2-w_n*|u(s)|^2||_\infty
\end{equation}
and note that $\lim_{n\to\infty}\epsilon_n=0$ by \eqref{eq:|||u|^2-w_n*|u|^2|| converges to 0}, where we use the dominated convergence theorem with $2\sup_{s\in[0,T]}||u(s)||_\infty<\infty$ as a dominating function. By Gronwall's inequality, we get for every $t\in[0,T]$
\begin{equation}
||(u-u_n)(t)||_2^2\le \epsilon_ne^{Ct}\, , 
\end{equation}
and therefore
\begin{equation}
\lim_{n\to\infty} ||(u-u_n)(t)||_2^2=0\, .
\end{equation}
We get pointwise almost everywhere convergence to $u$ at least for a subsequence $(u_{n_k})_{k\in\N}$. It follows that we get the same dispersive estimate for $u$:
\begin{equation}\label{eq:dispersive estimate for u in the proof of the cubic nls as a limit}
||u(t)||_\infty\le \frac{C_0}{(1+|t|)^{\frac{d}{2}}}
\end{equation}
for all $t\in[0,T]$. Now the blow-up alternative implies that $[0,\infty)\subset I$ and therefore, we get \eqref{eq:dispersive estimate for u in the proof of the cubic nls as a limit} for every $t\ge 0$. 
\\ \\
The proof of the dispersive estimate for $\partial_t u$ follows from the pointwise almost everywhere convergence of a subsequence of $(\partial_t u_n)(t)$ to $(\partial_t u)(t)$ for every $t\ge 0$. For the term with $-\Delta u$, we use the compact Sobolev embedding of $H^{k-2}(\R^d)$ into $L^2_{loc}(\R^d)$ for $k>2$. Thus, we need $k=4$ in dimension $d=3$ for this proof strategy. We omit the details here. 
\end{proof}
\printbibliography
\end{document}